\documentclass[conference,letterpaper]{IEEEtran}

\usepackage[utf8]{inputenc} 
\usepackage[T1]{fontenc}
\usepackage{url}
\usepackage{ifthen}
\usepackage{cite}
\usepackage[cmex10]{amsmath} % Use the [cmex10] option to ensure complicance
\usepackage{multirow}
\IEEEoverridecommandlockouts
\usepackage{cite}
\usepackage{amsmath,amssymb,amsfonts}

\usepackage{algorithm,algpseudocode}
\usepackage[font=small]{caption}
\usepackage[dvips]{graphicx}
\usepackage{colortbl}
\usepackage{array, tabularx}
\usepackage{subfigure}
\usepackage{amsthm}
\usepackage{mathtools}
\usepackage{setspace}
\usepackage{soul}
\usepackage{bm}
\usepackage{bbm}
\usepackage{mathtools}
\usepackage[utf8]{inputenc}
\usepackage[english]{babel}
\usepackage{subfigure}

\newtheorem{theorem}{Theorem}
\newtheorem{lemma}{Lemma}
\newtheorem{corollary}{Corollary}

\newtheoremstyle{case}{}{}{}{}{}{:}{ }{}
\theoremstyle{case}

\DeclarePairedDelimiter{\floor}{\lfloor}{\rfloor}
\DeclarePairedDelimiter{\ceil}{\lceil}{\rceil}
\DeclarePairedDelimiter{\nint}\lfloor\rceil
\def \TM
{T_{\mathrm{comp}}(C_{\mathrm{MDS}}(n,k))}
\def \TG {T_{\mathrm{comp}}(C_{\mathrm{G}}(\vct{n}, \vct{k}))}
\def \TGO {T_{\mathrm{comp}}(C_{\mathrm{G}}(\vct{n}, \vct{k}^*))}

\newcommand{\Max}[1]{\raisebox{0.5ex}{\scalebox{0.8}{$\displaystyle \max_{#1}\;$}}}
\newcommand{\Lim}[1]{\raisebox{0.5ex}{\scalebox{0.8}{$\displaystyle \lim_{#1}\;$}}}
\newcommand{\mtx}{\mathbf}
\newcommand{\vct}{\bm}

\DeclarePairedDelimiter{\abs}{\lvert}{\rvert}

\begin{document}
\title{Coded Matrix Multiplication \\
	on a Group-Based Model} 

% %%% Single author, or several authors with same affiliation:
% \author{%
%   \IEEEauthorblockN{Stefan M.~Moser}
%   \IEEEauthorblockA{ETH Zürich\\
%                     ISI (D-ITET)\\
%                     CH-8092 Zürich, Switzerland\\
%                     Email: moser@isi.ee.ethz.ch}
% }

%%% Several authors with up to three affiliations:
\author{%
  \IEEEauthorblockN{Muah Kim}
  \IEEEauthorblockA{\textit{School of Electrical Engineering} \\
  	\textit{KAIST}\\
  	Daejeon, Republic of Korea \\
  	02mu-a21@kaist.ac.kr}
  \and
  \IEEEauthorblockN{Jy-yong Sohn}
  \IEEEauthorblockA{\textit{School of Electrical Engineering} \\
  	\textit{KAIST}\\
  	Daejeon, Republic of Korea  \\
  	jysohn1108@kaist.ac.kr\\}
  \and
  \IEEEauthorblockN{Jaekyun Moon}
  \IEEEauthorblockA{\textit{School of Electrical Engineering} \\
  	\textit{KAIST}\\
  	Daejeon, Republic of Korea  \\
  	jmoon@kaist.edu\\}
}

%%% Many authors with many affiliations:
% \author{%
%   \IEEEauthorblockN{Albus Dumbledore\IEEEauthorrefmark{1},
%                     Olympe Maxime\IEEEauthorrefmark{2},
%                     Stefan M.~Moser\IEEEauthorrefmark{3}\IEEEauthorrefmark{4},
%                     and Harry Potter\IEEEauthorrefmark{1}}
%   \IEEEauthorblockA{\IEEEauthorrefmark{1}%
%                     Hogwarts School of Witchcraft and Wizardry,
%                     1714 Hogsmeade, Scotland,
%                     \{dumbledore, potter\}@hogwarts.edu}
%   \IEEEauthorblockA{\IEEEauthorrefmark{2}%
%                     Beauxbatons Academy of Magic,
%                     1290 Pyrénées, France,
%                     maxime@beauxbatons.edu}
%   \IEEEauthorblockA{\IEEEauthorrefmark{3}%
%                     ETH Zürich, ISI (D-ITET), ETH Zentrum, 
%                     CH-8092 Zürich, Switzerland,
%                     moser@isi.ee.ethz.ch}
%   \IEEEauthorblockA{\IEEEauthorrefmark{4}%
%                     National Chiao Tung University (NCTU), 
%                     Hsinchu, Taiwan,
%                     moser@isi.ee.ethz.ch}
% }

\maketitle

%%%%%%
%% Abstract: 
%% If your paper is eligible for the student paper award, please add
%% the comment "THIS PAPER IS ELIGIBLE FOR THE STUDENT PAPER
%% AWARD." as a first line in the abstract. 
%% For the final version of the accepted paper, please do not forget
%% to remove this comment!
%%
\begin{abstract}
%\emph{THIS PAPER IS ELIGIBLE FOR THE STUDENT PAPER AWARD.} 
Coded distributed computing has been considered as a promising technique which makes large-scale systems robust to the ``straggler" workers. 
Yet, practical system models for distributed computing have not been available that reflect the clustered or
grouped structure of real-world computing servers. Neither the large variations in the computing power and bandwidth capabilities across different servers have been properly modeled.
We suggest a \emph{group-based model} to reflect practical conditions and develop an appropriate coding scheme for this model. The suggested code, called \emph{group code}, employs parallel encoding for each group. 
We show that the suggested coding scheme can asymptotically achieve optimal computing time in regimes of infinite $n$, the number of workers. While theoretical analysis is conducted in the asymptotic regime, numerical results also show that the suggested scheme achieves near-optimal computing time for any finite but reasonably large $n$. Moreover, we demonstrate that the decoding complexity of the suggested scheme is significantly reduced by the virtue of parallel decoding.
\end{abstract}

%% The paper must be self-contained. However, if you are referring to
%% a full version for checking certain proofs, please provide the
%% publically accessible location below.  If the paper is completely
%% self-contained, you can remove the following line from your
%% submission.

%\textit{A full version of this paper is accessible at:}
%\cmt{\url{http://isit2019.fr/}} 

\section{Introduction}
In the era of big data, 
%As the amount of data and the volume of data have increased, 
distributed computing has been recognized as a solution for realizing large-scale machine learning \cite{ref:largescale}. Unlike conventional centralized systems, a distributed computing system divides the computational work into subtasks and distributes them over multiple nodes. This system successfully supports large-scale machine learning by reducing the computing time via parallel computing.

Yet, there is still a room for improvement as the system is slowed down by the random nature of computing nodes, where certain nodes are inevitably slower than others. In particular, the distributed system is shown to be dramatically degraded by the slowest workers, the ``stragglers", whose computational latency is realized by the tail probability \cite{ref:tail}. \emph{Lee} et al. suggested coded computation as a straggler-proof scheme, which speeds up matrix multiplication by employing redundancy with a maximum distance separable (MDS) code \cite{ref:speedup}. Afterwards, it is shown that coded computation can effectively improve the performance of computing system with regards to: matrix-matrix multiplication \cite{ref:highd, ref:polycode, ref:Tavor}, distributed gradient descent \cite{ref:gradientcodingArxiv, ref:gradientdescentRS}, convolution \cite{ref:codedconvolution}, Fourier transform \cite{ref:codedFourier}, 
%personalized PageRank \cite{ref:personalizedpagerank}, 
and matrix sparsification \cite{ref:shortdot, ref:matrixsparsif}. 
%\cmt{add more}. 
Moreover, regarding the matrix multiplication, new models reflecting the practical environment of computing systems such as the tree structure and heterogeneity are suggested and analyzed \cite{ref:HCMM,ref:hier}.

In recent years, distributed cloud computing services such as Amazon EC2 enable customers to deal with large-scale computation \cite{ref:AmazonEC2}. The real distributed computing systems generally adopt the multi-rack structure, where the computing workers are grouped together in multiple racks \cite{ref:MapReduce, ref:ShuffleWatcher, ref:ScaleOut}. Moreover, in the real world, the workers' latency statistics are heterogeneous due to a mixed use of hardwares with varying performances or the dynamics of multiple user requests over shared resources \cite{ref:HeteroEnv}. So far, the homogeneous grouped structure has been considered in \cite{ref:hier}, and the heterogeneous workers without grouped feature has been studied in \cite{ref:HCMM}. 
However, system solutions which reflect both of the two practical conditions$-$grouped structure and heterogeneity (in terms of number of workers in each group as well as the bandwidth of the communication links associated with the groups)$-$are yet to be established. 
%Yet, a model considering both of the two practical conditions, grouped structure and heterogeneity, has not been developed. We come up with designing a new model fitting these conditions and finding out a coded computation scheme which is suitable for the heterogeneous grouped structure.

\begin{figure}[!t]
	\centering
	\includegraphics[height=50mm]{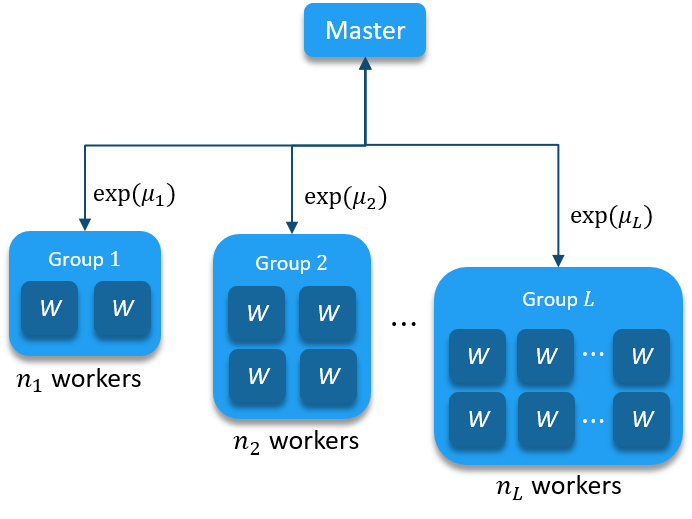}
	\caption{Computing Network Model: an $(\vct{n},\bm{\mu})-$group system with $L$ groups. Group $i$ has $n_i$ workers having i.i.d. completion time distribution with statistical parameter $\mu_i$ for $i\in[L]$.}
	\label{Fig:network1}
	\vspace{-6mm}
\end{figure}

\subsection{Main Contributions}

We design a group-based computing model as shown in Fig. \ref{Fig:network1}, where $n$ workers are dispersed into $L$ groups, each having a different number of nodes and distinct computing time statistics. We assume that group $i$ has $n_i$ nodes, each of which has a computing time given by an exponential random variable with rate $\mu_i$. 
This is a more practical model than the existing ones because it resembles the tree-shaped (grouped) distributed computing systems such as the Hadoop file system while also considering the heterogeneity of the groups.

Considering the scenario of computing $k$ tasks in the suggested model, we show that an $(n,k)-$MDS code achieves the optimal computing time. Yet, this scheme requires a prohibitive decoding complexity as $k$ increases. In addition, it is hard to obtain a closed-form expression for the optimal computing time due to the heterogeneous nature of the model.

To address these issues, we propose a coding scheme called \textit{group code} which divides the total $k$ tasks into $L$ partitions and then employs $L$ distinct MDS codes. We show that a carefully designed group code can asymptotically achieve the optimal computing time as $n$ goes to infinity. In addition, the suggested group code can reduce the decoding complexity down to a factor of $(\frac{1}{L})^{\beta}$ compared to an $(n,k)-$MDS code, where $\beta>1$.
Furthermore, we obtain a closed-form expression for the expected optimal computing time, when the number of workers $n$ goes to infinity. 

\subsection{Related Works}
Previous works on coded computation either achieves the optimal computing time with a prohibitive decoding complexity, or reduce the decoding complexity at the sacrifice of the optimality in computing time. In addition, most of them assume homogeneous workers. Applying an $(n,k)-$MDS code in homogeneous systems is suggested by \cite{ref:speedup}, which achieves the optimal computing time but requires a huge  decoding complexity as $k$ increases. 
Considering a system model with heterogeneous workers, the authors of \cite{ref:HCMM} suggested a coding scheme which achieves an asymptotically optimal computing time. However, the decoding process requires the computational complexity of $\mathcal{O}(k^3)$. Moreover, the coding schemes suggested in \cite{ref:highd, ref:hier, ref:Tavor} encode the tasks along multiple dimensions, which can
effectively reduce the decoding complexities by the virtue of parallel decoding or a peeling decoding scheme. However, these codes lose the MDS property and thereby cannot achieve the optimal computing time. Besides, 
these codes do not provide solutions for practical systems with heterogeneous groups. Compared to these existing works, our suggested scheme is shown to not only asymptotically achieve the optimal computing time, but also requires a low decoding complexity.

%Moreover, it is analyzed on a heterogeneous group-based model.
%have hardly looked into the effect of using multiple codes independently, which enables the parallel encoding and decoding. 
% Our suggested model is featured by its grouped structure and the heterogeneity among groups. The grouped structure is inspired by the tree-structures used in \cite{ref:ScaleOut, ref:server-rack, ref:hier}. Coded matrix multiplication on a similar hierarchical structure with homogeneous groups was considered in \cite{ref:hier}, while the present paper considers scenarios of heterogeneous groups. The heterogeneity in response time distribution is previously studied in \cite{ref:HCMM}, where all the workers have distinct distributions (i.e., each group has one worker). Compared to the work of \cite{ref:HCMM}, the present paper considers a more practical setting with grouped workers. 
% (mention  why the grouped model is more appropriate to reflect the practical distributed computing system)
%Moreover, motivated by the practical implementation issue, the encoding or decoding time of their coded distributed computing system is analyzed by the authors of \cite{ref:speedup, ref:Tavor, ref:highd, ref:hier}. The present paper also considers the encoding and decoding time in the overall latency analysis, motivated by the similar reasons.

\subsection{Notations}

Here, we list mathematical notations used in this paper. For a positive integer $n$, a set of positive integers less than or equal to $n$  is denoted by $[n]=\{1,2,\dots, n\}$. For a matrix $\mtx{A}$ with multiple rows, $\mtx{A}=[\mtx{A}_1 ; \mtx{A}_2]$ represents row-wise division of $\mtx{A}$, i.e. $\mtx{A}^T=[\mtx{A}_1^T  \mtx{A}_2^T]$. 
We use $C_G(\vct{n},\vct{k})$ to denote  an $(\vct{n},\vct{k})-$group code and $C_\mathrm{MDS}(n,k)$ to denote an $(n,k)-$MDS code. The definition of  $(\vct{n},\vct{k})-$group code is in Section \ref{Subsection:compmodel}. We denote the floor, ceil and round functions of a real value $x$ by $\floor{x},\ceil{x}$ and $\nint{x}$.  
% For $a, b \in \mathbb{Z}^+$, $a|b$ means $a$ can divide $b$. Moreover, for two random variables $A$ and $B$, $A \xrightarrow{d} B$ and $A \xrightarrow{p} B$ represent the convergence of $A$ to $B$ in distribution and in probability, respectively. We shorten a sequence $x_1 , x_2, \dots, x_n$ as $x_1^n$. 
% For proof: For a probability distribution $P(\cdot)$, the empirical distribution of $P(\cdot)$ with $n$ samples is represented by $\widetilde{P}_n(\cdot)$. Moreover, we denote the cumulative distribution function of the standard normal distribution by $\Phi(\cdot)$. 

\begin{figure}[!t]
	\centering
	\includegraphics[height=40mm]{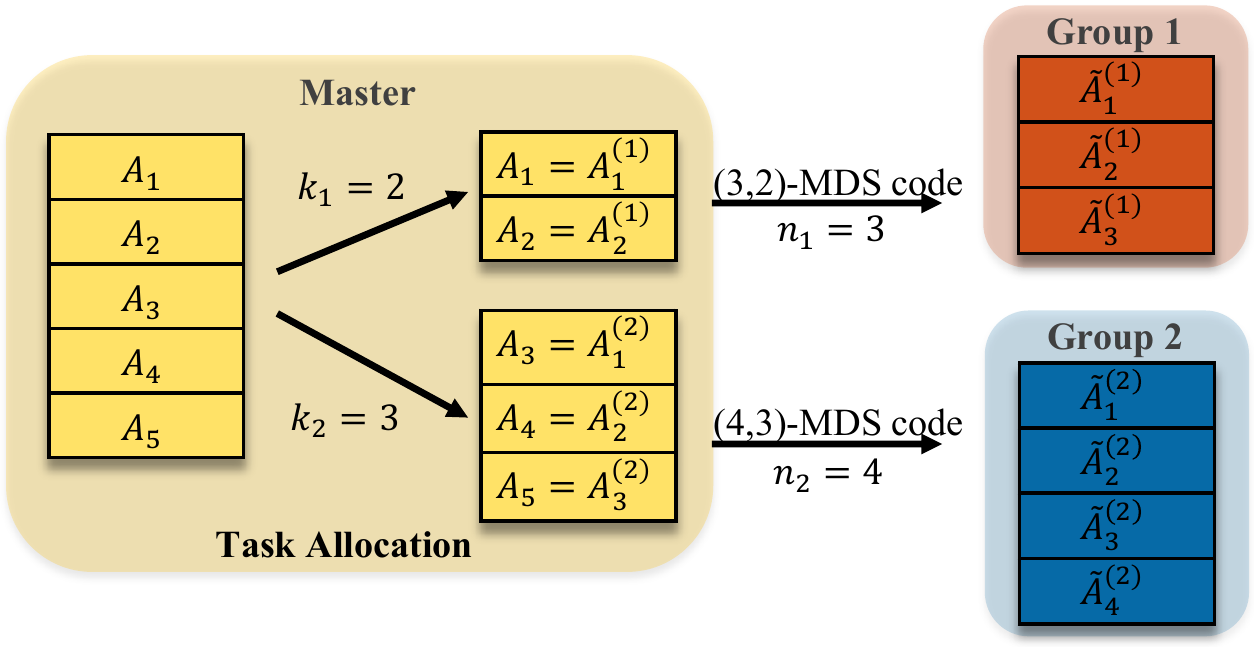}
	\caption{Illustration of $(\vct{n}, \vct{k})=([3,4],[2,3])-$group code. The matrix is split into two submatrices following the task allocation vector $\bm{k}=[2,3]$, and then the submatrices are encoded with MDS code group-wise.}
	\label{Fig:comp}
%	\vspace{-5mm}
\end{figure}

\section{System Model and Target Problem}\label{Section:SysMod}

\subsection{System Model}\label{Subsection:compmodel}

Consider the $n$ workers that are spread into $L$ groups as shown in Fig. \ref{Fig:network1}. Here, group $i$ has $n_i$ workers whose response times are described by i.i.d. random variables with a parameter of $\mu_i$. We define this system as an $(\bm{n},\bm{\mu})-$\textit{group system}, where $\vct{n}=[n_1,n_2,\dots,n_L]$ and $\bm{\mu}=[\mu_1,\mu_2,\dots,\mu_L]$. For simplicity, we call the $j^{th}$ worker in group $i$ as $w(i,j)$ for $i\in[L]$ and $j\in[n_i]$.
We implement a matrix-vector multiplication $\mtx{A}\vct{x}$ on this system, where $\mtx{A} \in \mathbb{R}^{m \times d}$ is a work matrix and  $\vct{x} \in \mathbb{R}^{d \times 1}$ is an input vector for some positive integers $m$ and $d$. Now, the work matrix $\mtx{A}$ is divided into equal-sized $k$ submatrices as $\mtx{A}=[\mtx{A}_1 ; \mtx{A}_2 ; \cdots ; \mtx{A}_k ]$, where $k$ is a positive integer that can divide $m$, and $\mtx{A}_{r} \in \mathbb{R}^{\frac{m}{k}\times d}$ for $r \in [k]$. 

The task of computing $\mtx{A}\vct{x}$ is distributed to $n$ workers as below. First, we define $\vct{k}=[k_1, k_2, \dots, k_L]$ as a \textit{task allocation} vector, where the elements are positive integers satisfying $\sum_{i=1}^L k_i=k$. The set of submatrices $\{\mtx{A}_r\}_{r=1}^k$ is now partitioned into $L$ disjoint subsets $\{\mathbb{S}_i\}_{i=1}^L$ such that $\abs{\mathbb{S}_i}=k_i$ holds for $i\in [L]$. We denote the elements in set $\mathbb{S}_i$ as $\mathbb{S}_i=\{\mtx{A}^{(i)}_{j}\}_{j=1}^{k_i}.$ Afterwards, the $k_i$ elements of $\mathbb{S}_i$ are encoded with an $(n_i,k_i)-$MDS code and we denote the set of $n_i$ coded submatrices by $\widetilde{\mathbb{S}}_i=\{\widetilde{\mtx{A}}_j^{(i)}\}_{j=1}^{n_i}$. Worker $w(i,j)$ now stores $\widetilde{\mtx{A}}_j^{(i)}$ and computes $\widetilde{\mtx{A}}_j^{(i)}\vct{x}$ when it receives the input vector $\vct{x}$ from the master. We call this coding scheme as an $(\vct{n},\vct{k})-$\textit{group code}, denoted by $C_\mathrm{G}(\vct{n},\vct{k})$. Fig. \ref{Fig:comp} illustrates an example of an $(\vct{n},\vct{k})-$group code when $\bm{n}=[3,4]$ and $\bm{k}=[2,3]$. The matrix $\mtx{A}=[\mtx{A}_1;\mtx{A}_2;\dots;\mtx{A}_5]$ is divided into two sets of submatrices, $\{\mtx{A}_1,\mtx{A}_2\}$ and $\{\mtx{A}_3,\mtx{A}_4,\mtx{A}_5\}$. Then, by applying a $(3,2)-$MDS code and a $(4,3)-$MDS code, respectively, we obtain $\{\widetilde{\mtx{A}}_1^{(1)},\widetilde{\mtx{A}}_2^{(1)},\widetilde{\mtx{A}}_3^{(1)}\}$ and $\{\widetilde{\mtx{A}}_1^{(2)},\widetilde{\mtx{A}}_2^{(2)},\widetilde{\mtx{A}}_3^{(2)},\widetilde{\mtx{A}}_4^{(2)}\}$. 
Each worker individually transmits its computational result $\widetilde{\mtx{A}}_j^{(i)}\vct{x}$  to the master when its computation is finished. To obtain the computational output $\mtx{A}\vct{x}$, the master needs at least $k_i$ computational results from each group $i$ to decode the $(n_i,k_i)-$MDS code. Note that this model can be directly applied to the matrix-matrix multiplication, where the input vector $\vct{x}$ is replaced by a matrix $\mtx{B} \in \mathbb{R}^{d \times c}$. 

We adopt the exponential distribution model for the \emph{completion time} of a worker, which is defined as the time taken for both the computation and the transmission of the computed result to the master. This model has also been assumed in other papers on coded computation \cite{ref:highd,ref:hier}. 
Unlike these papers, however, a worker in group $i$ has the distribution parameter of $\mu_i$, where $\mu_i$ varies among different groups. More precisely, the completion time $T_j^{(i)}$ of worker $w(i,j)$ is defined by its cumulative distribution function as
%\begin{equation*}
$\Pr[T_j^{(i)} \leq t] = 1-e^{k\mu_i  t}$
%\end{equation*}
for time $t \geq 0$. Here, the completion time has the rate of $k\mu_i$ since the number of rows in the submatrix $\mtx{A}_r \in \mathbb{R}^{\frac{m}{k} \times d}$ becomes smaller as $k$ increases.
%In accordance with the recent papers on distributed computing systems \cite{ref:highd},\cite{ref:hier}, the completion time of each worker is modeled as an exponential distribution. we adopt the exponential distribution following the straggler model of related papers \cite{ref:highd,ref:hier} on distributed computing systems.
%since it is well extended to general response time distributions as well as its simplicity. \cmt{paraphrase} \cmt{Why use exponential r.v.? - high-d:"We remark that the analysis can be easily extended to general response time distributions."} 

\subsection{Target Problem}
This paper mainly aims at analyzing the total execution time $T_{\mathrm{exec}}$ of $(\vct{n},\vct{k})-$group codes, which refers to the entire time taken for computing and decoding. The computing time $T_{\mathrm{comp}}$ is the time taken for the master to gather computational subtasks from the workers, while the decoding time $T_{\mathrm{dec}}$ is the time taken to recover the original task of computing $\mtx{A}\vct{x}$ from the gathered subtasks. In this paper, we assume that the encoding time complexity is negligible compared to $T_{\mathrm{comp}}$ and $T_{\mathrm{dec}}$. This is because we focus on the scenarios of multiplying varying input vectors with the same work matrix $\mtx{A}$, which is encoded once prior to the computation. Thus, we have 
\begin{equation*}
T_{\mathrm{exec}}(C)=T_{\mathrm{comp}}(C) +T_{\mathrm{dec}}(C),
\end{equation*}	
when code $C$ is applied to the system.
%Similar notations can be used for $T$ and $T_\mathrm{exec}$, respectively.

We focus on analyzing the computing time of $(\vct{n},\vct{k})-$group codes, which is denoted by $\TG$. Recall that the computing time of an $(\vct{n},\vct{k})-$group code is equivalent to the time when every group $i$ has at least $k_i$ workers which finish their tasks.
Let $T_{k_i:n_i}^{(i)}$ be the $k_i^{th}$ smallest value among $\{T_j^{(i)}\}_{j=1}^{n_i}$. Then, $\TG$ can be expressed as
\begin{equation*}
\TG = \max(T_{k_1:n_1}^{(1)}, T_{k_2:n_2}^{(2)}, \dots, T_{k_L:n_L}^{(L)}).
\end{equation*}
Since it is hard to find a closed-form expression for $\mathbb{E}[\TG]$ when $n$ is finite, 
%know the exact distribution of the maximum among multiple order statistics, 
we set our main problem as to obtain the expected value as $n$ goes to infinity, i.e.
\begin{equation*}\label{Eqn:MainProb}
\mathbb{P}_\mathrm{main}:	\mathrm{compute}	\lim_{n\to\infty} \mathbb{E}[\TG]. 
\end{equation*}
Here, we assume $k=\Theta(n)$ and $n_i=\Theta(n)$ for $i\in[L]$. 
%throughout the paper. 

\section{Optimal Computing Time Analysis
	%:\\ $(n,k)-$MDS code
}	

Here we find the optimal computing time of a given $(\vct{n}, \vct{\mu})-$group system. Theorem \ref{Thm:MDS_OPT} states that applying an $(n,k)-$MDS code achieves the optimal computing time.
%, where the MDS code is applied to the system as follows.
We consider an $(n,k)-$MDS code is applied to the $k$ submatrices $\{\mtx{A}_1, \mtx{A}_2, \dots, \mtx{A}_k\}$, resulting in $n$ coded submatrices $\{\widetilde{\mtx{A}}_1, \widetilde{\mtx{A}}_2, \dots, \widetilde{\mtx{A}}_n\}$. Then, the $n$ coded submatrices are distributed to $n$ workers regardless of the groups they belong. Here we denote the computing time of an $(n,k)-$MDS code as $\TM$. %Then, $\TM$ can be represented as the $k^{th}$ order statistic of total $n$ random variables, denoted as
%	\begin{equation*}
%	T_{\mathrm{comp\_MDS}}(\vct{n},k,\bm{\mu})=T_{k:n}.
%	\end{equation*}
%	
%Then, the following theorem shows that $\TM$ is the optimal computing time of $(\vct{n}, \vct{\mu})-$group model.
\begin{theorem}\label{Thm:MDS_OPT}
	Consider computing $k$ tasks on $(\vct{n}, \vct{\mu})-$group systems. 
	Then, an $(n,k)-$MDS code achieves the optimal computing time. In other words,  for arbitrary $(n,k)$ linear code $C \in \mathcal{C}(n,k)$, 
	\begin{align*}
	\TM \leq T_{\mathrm{comp}}(C).
	\end{align*}
\end{theorem}	
%The proof of lemma 
\begin{proof}
	Given an arbitrary realization of the completion times  $\{T_j^{(i)}\}_{i\in[L],j\in[n_i]}$ of workers, we can think of their order statistics $T_{1:n}<T_{2:n}<\dots<T_{n:n}$.
	%A linear code of length $n$ and rank $k$ whose minimum distance is $d$ can correct upto $d-1$ erasures. By Singleton bound, $d \leq n-k+1$. 
	Recall that $(n,k)$ linear code $C$ cannot recover the original message if there are more than $n-k$ erasures, which leads to $T_{\mathrm{comp}}(C)\geq T_{k:n}$. By the MDS property, we have $ \TM=T_{k:n}$, which completes the proof.
\end{proof}

\section{Computing Time Analysis}\label{section:MainRes}	

In this section, we provide the computing time analysis when the workers are dispersed into $L=2$ groups. 

\subsection{Computing Time for an Arbitrary Task Allocation $\bm{k}$}\label{Subsection:CompTimeL2}	

For simplicity, we denote the task allocation vector as  $\bm{k}=[k_1,k_2]=[k_1, k-k_1]$.
The computing time of an $(\bm{n},\bm{k})-$group code for $L=2$ can be expressed as  $\TG=\max(T^{(1)}_{k_1:n_1}, T^{(2)}_{k_2:n_2})$ by definition. Lemma \ref{Lemma:OrdBehavior} provides 	
the expected computing time of an $(\bm{n},\bm{k})-$group code when $n$ goes to infinity.
\begin{lemma} \label{Lemma:OrdBehavior}
	Consider an $(\bm{n},\bm{\mu})-$group system with $L=2$ groups. Then, the expected computing time of an $(\bm{n},\bm{k})-$group code satisfies the following:
	\begin{align}
	\lim_{n\to\infty} &\mathbb{E}[\TG]=\lim_{n\to\infty} \mathbb{E}[\max(T_{k_1:n_1}^{(1)},T_{k_2:n_2}^{(2)})]\nonumber\\
	=&\max(\lim_{n\to\infty}\mathbb{E}[T_{k_1:n_1}^{(1)}],\lim_{n\to\infty}\mathbb{E}[ T_{k_2:n_2}^{(2)}])\nonumber \\
	=&\max\left(-\dfrac{1}{k\mu_1}\log(1-\dfrac{k_1}{n_1}),-\dfrac{1}{k\mu_2}\log(1-\dfrac{k_2}{n_2})\right)\label{Eqn:maxTerm} .
	\end{align} 
\end{lemma}
\begin{proof}
	We set aside the proof at Appendix \ref{proof:lemmel2}.
\end{proof}
This lemma illustrates that in the asymptotic regime of large $n$, the expected computing time of an $(\bm{n},\bm{k})-$group code can be easily obtained for given $\vct{n}$, $\vct{\mu}$ and $k$.

%This lemma illustrates that in the asymptotic regime of large $n$, the expected computing time of an $(\bm{n},\bm{k})-$group code is determined by the expected time at which the slower group finishes its tasks.

%the converges to maximum of their expectation values.
%This lemma indicates that, in asymptotic region of large $n$, the computing time can be analyzed by the time taken for computing $k_1^*$ tasks in group $1$, which is the $k_1^{*th}$ order statistic of $n_1$ i.i.d. random variables.
Now, we aim at optimizing task allocation rule $\bm{k}$ which minimizes the computing time of an $(\vct{n}, \vct{k})-$group code. We define the optimal task allocation vector by
\begin{equation}\label{Eqn:kStar}
\bm{k}^* \coloneqq \underset{\bm{k}}{\arg\min} \mathbb{E}[\TG],
\end{equation}
whose elements are denoted by $\bm{k}^*=[k_1^*,k_2^*,\dots,k_L^*]$. 
Before finding $\vct{k}^*$, 
%Before finding $k^*$ in Theorem \ref{Lemma:ExecTime_L2}, 
we state a relationship between $T_{k:n}$ and $\{T_{k_i:n_i}^{(i)}\}_{i=1}^L$ in the following Lemma. Recall that $T_{k:n}$ is equivalent to $\TM$, and the maximum among $\{T_{k_i:n_i}^{(i)}\}_{i=1}^L$ corresponds to $\TG$ by definition.

%We begin with finding upper and lower bounds on $T_{k:n}=\TM$ as in Proposition \ref{proposition:Bounds_L2}. For simplicity, we eliminate $k_2$ by replacing it with $k-k_1$ in this subsection.
\begin{lemma}\label{proposition:Bounds_L2}
	Under the scenario of computing $k$ tasks on an $(\bm{n},\bm{\mu})-$group system with $L=2$ groups, consider applying an $(\bm{n},\bm{k})-$group code where $\bm{n}=[n_1,n_2]$ and $\bm{k}=[k_1,k-k_1]$.
	Given an arbitrary realization of completion time $\{T_j^{(i)}\}_{i \in [2], j \in [n_i]}$ of workers, let $T_{k:n}$ be the $k^{th}$ smallest value among $\{T_j^{(i)}\}_{i\in[2],j\in[n_i]}$. Meanwhile, $T_{k_i:n_i}^{(i)}$ denotes the $k_i^{th}$ smallest value among $\{T_j^{(i)}\}_{j=1}^{n_i}$. Then, we have
	\begin{equation}\label{Eqn:prop1}
	\min(T_{k_{1}:n_{1}}^{(1)},T_{k-k_{1}:n_{2}}^{(2)}) \hspace{-1mm}\le T_{k:n} \hspace{-1mm}\le \max(T_{k_{1}:n_{1}}^{(1)},T_{k-k_{1}:n_{2}}^{(2)}).
	\end{equation}
\end{lemma}
\begin{proof}
	%By definition, we only have $k_1 \in (\max(0,k-n_2),\min (k,n_1))$ because $k_1\leq n_1$ and $k-k_1\leq n_2$
	%Among the $k$ tasks finished earliest, let $k_1'$ be the number of tasks completed by workers in group $1$. Accordingly, the number of tasks done by workers in group $2$ is $k-k_1'$. Then, we may write $T_{k-k_1'-1:n_2}^{(2)}<T_{k:n}<T_{k_1'+1:n_1}^{(1)}$. 
	Let 
	%$\mathbb{U}$ be a set of completion times such that 
	$\mathbb{U}=\{T_{j}^{(i)}:T_{j}^{(i)}\leq T_{k:n} \mathrm{ for } i\in[2], j\in[n_i]\}$. Consider a subset $\mathbb{U}_1$ of set $\mathbb{U}$ such that $\mathbb{U}_1=\{T_{j}^{(1)}:T_{j}^{(1)}\le T_{k:n} \mathrm{ for }  j\in[n_1]\}$ and its complementary set $\mathbb{U}_1^C=\{T_{j}^{(2)}:T_{j}^{(2)}\le T_{k:n} \mathrm{ for }  j\in[n_2]\}$. 
	%Similarly, we define $\mathbb{U}_2= \{T_{j}^{(2)}:T_{j}^{(2)}\le T_{k:n} \mathrm{for}  j\in[n_i]\}$. 
	Here, we define $k_1'\coloneqq\abs{\mathbb{U}_1}$. Notice that $\abs{\mathbb{U}_1^C}=k-k_1'$. Then, we may write $T_{k-k_1'-1:n_2}^{(2)}<T_{k:n}<T_{k_1'+1:n_1}^{(1)}$.
	When $k_1'<k_1$, 
	%$T_{k_{1}:n_{1}}^{(1)}$ is greater than or equal to $T_{k:n}$ because 
	we have $T_{k:n}<T_{k_1'+1:n_1}^{(1)}\leq T_{k_1:n_1}^{(1)}$. Similarly, we have $T_{k:n}>T_{k-k_1'-1:n_2}^{(2)}\geq T_{k-k_1:n_2}^{(2)}$, which leads to $ T_{k-k_1:n_2}^{(2)}\leq T_{k:n}\leq T_{k_{1}:n_{1}}^{(1)}$. When $k_1'>k_1$, we have $T_{k_{1}:n_{1}}^{(1)} \leq T_{k:n}\leq T_{k-k_1:n_2}^{(2)}$ using the same method as above. For $k_1'=k_1$, it is obvious that	$\min(T_{k_{1}:n_{1}}^{(1)},T_{k-k_{1}:n_{2}}^{(2)}) < T_{k:n} = \max(T_{k_{1}:n_{1}}^{(1)},T_{k-k_{1}:n_{2}}^{(2)})$. This completes the proof.\qedhere		
\end{proof}
%\vspace{-1mm}
In the following theorem, we find the optimal task allocation $\vct{k}^*$, and show that the expected computing time of an $(\bm{n},\bm{k}^*)-$group code converges to that of an $(n,k)-$MDS code for sufficiently large $n$. 
% which is shown to be the optimal scheme in the sense of computing time in Theorem \ref{Thm:MDS_OPT}.

\begin{theorem} \label{Lemma:ExecTime_L2}
	Consider a scenario of computing $k$ tasks on an $(\bm{n},\bm{\mu})-$group system with $L=2$ groups, where an $(\bm{n},\bm{k})-$group code is applied.
	In the asymptotic regime of large $n$, the optimal task allocation $\bm{k}^*=[k_1^*,k-k_1^*]$ can be obtained\footnote{Here we assume that $k_1^*$ is an integer since the task allocation vector $\vct{k}$ consists of integers. However, in case of $k_1^*$ not an integer, the optimal allocation rule is either $\vct{k}=[\ceil{k_1^*}, k-\ceil{k_1^*}]$ or $\vct{k}=[\floor{k_1^*}, k-\floor{k_1^*}]$, since $\mathbb{E}[\TG]$ is a convex function of $k_1$, as in the proof.}
	by solving
	%is a real\footnote{It is safe to extend $k_1$ from a natural number to a positive real number when we deal with order statistics because an order statistic is a random variable defined for a real number $r\in(0,1)$ and $r^{th}$ population quantile $\xi_{r}$. Although $k^{th}$ order statistic out of $n$ random variables is defined under the condition that $k/n\to r$ as $n\to\infty$ in general, $k$ does not have to be confined to being a natural number.} vector satisfying 
	%	the following equation:
	\begin{equation}\label{eqn:OptKL2}
	%		\lim_{n\to\infty} \mathbb{E}[T_{k_1^*:n_1}^{(1)}] = \lim_{n\to\infty} \mathbb{E}[T_{k-k_1^*:n_2}^{(2)}].
	%\frac{1}{\mu_1}\log{\left(1-\frac{k_1^*}{n_1}\right)}=\frac{1}{\mu_2}\log{\left(1-\frac{k-k_1^*}{n_2}\right)}
	k_1^* + n_2 -  n_2 \left(1-\frac{k_1^*}{n_1}\right)^{\frac{\mu_2}{\mu_1}} =k.
	\end{equation}
	Moreover, the expected computing time of an $(\bm{n},\bm{k}^*)-$group code satisfies the following:
	\begin{equation}\label{Eqn:TGOTM}	
	\lim_{n\to\infty} \mathbb{E} [\TGO]=\hspace{-1mm}\lim_{n\to\infty} \mathbb{E} [\TM].
	\end{equation} 
	
\end{theorem}
\begin{proof}
	%The formal proof is written at Appendix \ref{proof:lemma:ExecTime_L}, but here we provide the sketch of the proof. 
	Combining \eqref{Eqn:maxTerm} and \eqref{Eqn:kStar}, we obtain
	\begin{align*}
	\lim_{n\to\infty} k_1^*
	&=\underset{k_1 \in [k]}{\arg\min} \Big\{\lim_{n\to\infty}\max(\mathbb{E}[T_{k_1:n_1}^{(1)}],\mathbb{E}[ T_{k-k_1:n_2}^{(2)}])\Big\}\\ 
	&=\underset{k_1 \in [k]}{\arg\min}\Big\{\max\Big(-\dfrac{1}{k\mu_1}\log(1-\dfrac{k_1}{n_1}),\\
	&\hspace{30mm}-\dfrac{1}{k\mu_2}\log(1-\dfrac{k-k_1}{n_2})\Big)\Big\}.
	\end{align*}
	Note that the first variable of the max function is a strictly increasing convex function of $k_1$, while the second one is a strictly decreasing convex function. Thus, taking the maximum of the two variables results in a convex function of $k_1$. Therefore, as $n$ grows to infinity, the minimizer $k_1^*$ coincides with the intersection point of the two functions, 
	i.e.,
	\begin{equation}
	\lim_{n\to\infty}\mathbb{E}[T_{k_1^*:n_1}^{(1)}]=\lim_{n\to\infty}\mathbb{E}[ T_{k-k_1^*:n_2}^{(2)}].\label{Eqn:limExp}
	\end{equation}
	From \eqref{Eqn:maxTerm} and \eqref{Eqn:limExp}, we obtain \eqref{eqn:OptKL2} by simple algebraic operations.
	Now we move on to the proof of \eqref{Eqn:TGOTM}. First, by taking $\Lim{n \to \infty}\mathbb{E}[\cdot]$ on \eqref{Eqn:prop1} and applying Lemma \ref{Lemma:OrdBehavior}, we obtain
	\begin{align*}
	\min \Big( \lim_{n\to\infty}\mathbb{E}[ T_{k_1:n_1}^{(1)}]& ,\lim_{n\to\infty}\mathbb{E}[ T_{k-k_1:n_2}^{(2)}] \Big) \leq \lim_{n\to\infty} \mathbb{E} [T_{k:n}] \\
	\le \max &\left(\lim_{n\to\infty}\mathbb{E}[ T_{k_1:n_1}^{(1)}],\lim_{n\to\infty}\mathbb{E}[T_{k-k_1:n_2}^{(2)}]\right).
	\end{align*}
	When $k_1 = k_1^*$, the upper and lower bounds have the same value as in \eqref{Eqn:limExp}. Thus, by squeeze theorem, we have
	\begin{align*}
	\lim_{n\to\infty} \mathbb{E} [T_{k:n}]=\lim_{n\to\infty}\mathbb{E}[ T_{k_1^*:n_1}^{(1)}]=\lim_{n\to\infty}\mathbb{E}[ T_{k-k_1^*:n_2}^{(2)}].
	\end{align*}
	Therefore, we obtain \eqref{Eqn:TGOTM} by using $\TM=T_{k:n}$ and $\TG=\Max{i \in [L]} T_{k_i:n_i}^{(i)}$.
\end{proof}

Recall that an $(n,k)-$MDS code achieves the optimal computing time as stated in Theorem \ref{Thm:MDS_OPT}. The above theorem implies that an $(\vct{n}, \vct{k})-$group coded system can asymptotically achieve the optimal computing time by using the optimal task allocation rule $\vct{k}=\vct{k}^*$.
%be minimized down to that of $(n,k)-$MDS code 
%by carefully allocating tasks over groups according to the given resources $\vct{n}$ and $\vct{\mu}$, i.e. 
Note that \eqref{eqn:OptKL2} can be easily solved when $\mu_1/\mu_2=2$ by using the quadratic formula. The following corollary provides the optimal task allocation $\bm{k}^*$ and the corresponding $\mathbb{E}[\TGO]$ when $\mu_1 = 2\mu_2$.
\begin{corollary}
	Consider the scenario of computing $k$ tasks on an $(\bm{n},\bm{\mu})-$group system with $L=2$ and $\bm{\mu}=[2\mu_2,\mu_2]$. 	
	Under the scenario of applying an $(\bm{n},\bm{k})-$group code on this system, the optimal task allocation $\bm{k}^*=[k_1^*, k-k_1^*]$ is obtained as   
	\begin{equation}\label{Corollary:statement1}
	k_1^*= k-n_2-\frac{n_2^2}{2n_1}+\sqrt{(n_2+\frac{n_2^2}{2n_1})^2-\frac{k}{n_1}n_2^2}.
	\end{equation}
	Moreover, the expected value of the corresponding computing time $\mathbb{E}[\TGO]$ can be calculated as 
	\begin{align}
	\lim_{n\to\infty} &\mathbb{E}[\TGO]\nonumber\\ 
	&= \frac{1}{k\mu_2}\log{\left(\sqrt{(1+\frac{n_2}{2n_1})^2-\frac{k}{n_1}}-\frac{n_2}{2n_1}\right)}^{-1}. \label{Eqn:LETGO}
	\end{align}
\end{corollary}
\begin{proof}
	When $\mu_1 = 2\mu_2$, the equation \eqref{eqn:OptKL2} reduces to \eqref{Corollary:statement1}. In addition, inserting \eqref{Corollary:statement1} into \eqref{Eqn:maxTerm} results in \eqref{Eqn:LETGO}.
\end{proof}

\subsection{Numerical Results when the Number of Nodes are Finite}

\begin{figure}
	\centering
	\includegraphics[width=55mm]{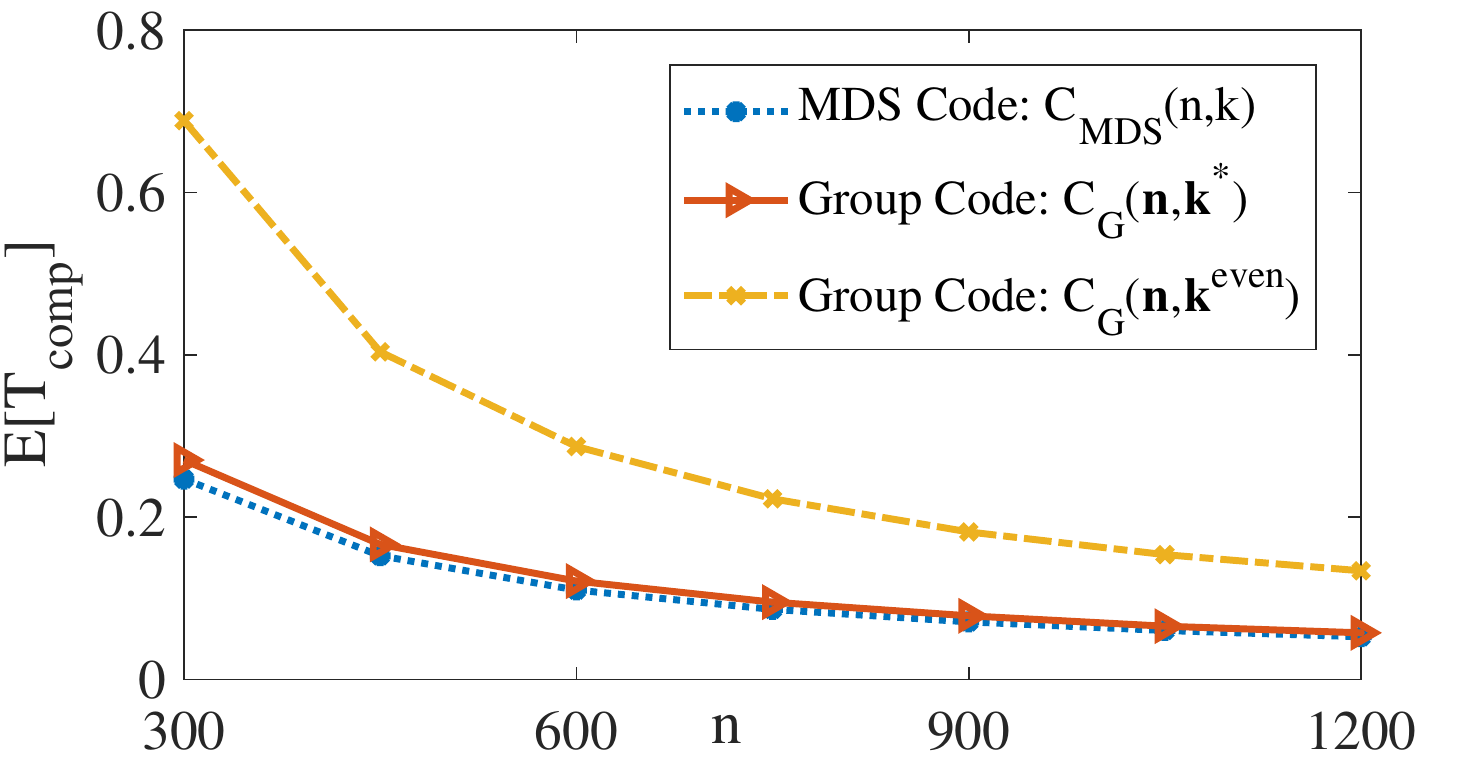}
	\caption{Simulated average computing time $\mathbb{E}[T_{\mathrm{comp}}]$ of an MDS code and two types of group codes. Parameters are set to $(\vct{n}, \vct{\mu})=\left([\frac{3}{4}n,\frac{1}{4}n],[1,2]\right)$ and $k=100$.}
	\label{Fig:simul_optimal_k*}
\end{figure}

Here, we provide simulation results on the computing time of an $(\vct{n}, \vct{k})-$group code when the number of nodes $n$ is finite. 
%The mathematical analysis conducted in previous subsections consider the asymptotic regime of large $n$, whereas 
%the following numerical results show that the suggested optimal task assignment rule $\vct{k}^*$ has near-optimal performance for finite $n$.
Fig. \ref{Fig:simul_optimal_k*} illustrates the expected computing time of an $(n,k)-$MDS code $\mathbb{E}[\TM]$ and that of $(\bm{n},\bm{k})-$group code $\mathbb{E}[\TG]$, for various $n$. We consider two types of group codes: one with the optimal task allocation $\vct{k}^*=[k_1^*, k-k_1^*]$, and the other with an even task allocation $\vct{k}^{\mathrm{even}}=[\frac{1}{2}k, \frac{1}{2}k]$.
%$C_\mathrm{MDS}(n,k)$  $C_G(\vct{n}, \vct{k})$
For a fixed number of tasks $k=100$, we assume that $n$ workers are divided into two groups as $\vct{n}=[n_1, n_2]=[\frac{3}{4}n,\frac{1}{4}n]$. Moreover, the average computing time of a worker doubles in the first group, i.e., $\vct{\mu}=[\mu_1, \mu_2] = [1,2]$. For the estimation, we employ Monte Carlo methods with $10^4$ random samples.
The simulation result demonstrates that the expected computing time of an $(\bm{n},\bm{k}^*)-$group code approaches to that of an $(n,k)-$MDS code in the asymptotic regime of large $n$, as proved in Theorem \ref{Lemma:ExecTime_L2}.
Moreover, the average computing times of two group codes $-$ the optimal group code $C_\mathrm{G}(\vct{n},\vct{k}^*)$ and a naive group code $C_\mathrm{G}(\vct{n},\vct{k}^{\mathrm{even}})$ $-$ have a significant gap, which supports the necessity of a careful task allocation considering the heterogeneity of groups.
% in order to reduce the expected computation time of a group code.
%Moreover, the simulation results for two group codes, $C_\mathrm{G}(\vct{n}, \vct{k}^*)$ and $C_\mathrm{G}(\vct{n}, \vct{k}^{\mathrm{even}})$, show that the reduction of 
%\cmt{The yellow dashed line, $\mathbb{E}[\TGE]$, indicates that the computing time can be significantly degraded when the equal amount of tasks are assigned to the distinct groups without considering their different performances.}

%\begin{figure}[!t]
%	\centering
%	\includegraphics[width=55mm]{simul_opt_k.pdf}
%	\caption{The expected computing time $\mathbb{E}[T_{\mathrm{comp}}]$ of an MDS code $C_{\mathrm{MDS}}(n,k)$ and two types of group codes $C_\mathrm{G}(\vct{n},\vct{k}^*)$, $C_\mathrm{G}(\vct{n},\vct{k}^{\mathrm{even}})$. Parameters are set to $(\vct{n}, \vct{\mu})=\left([\frac{3}{4}n,\frac{1}{4}n],[1,2]\right)$ and $k=100$.}
%	\label{Fig:simul_optimal_k*}
%\end{figure}

\section{Computing Time Analysis for General $L$}\label{Section:generalL}
\subsection{Computing Time for an Arbitrary Task Allocation $\bm{k}$}
This section provides the expected computing time of an $(\bm{n},\bm{k})-$group code for an arbitrary number of groups, i.e. $L\geq2$. The following lemma provides a numerical way to obtain $\TG$ as $n$ grows to infinity. 

\begin{lemma}\label{Lemma:LEMMELL} 
	Consider an $(\bm{n},\bm{\mu})-$group system with $L$ groups. Then, the expected computing time of an $(\bm{n},\bm{k})-$group code satisfies the following:
	\begin{align}
	\lim_{n\to\infty} \mathbb{E}[\TG]&=\lim_{n\to\infty} \mathbb{E}[\max_{i\in[L]} T^{(i)}_{k_i:n_i}] \nonumber \\
	=\max_{i\in[L]}& (\lim_{n\to\infty}\mathbb{E}[T_{k_i:n_i}^{(i)}]) \nonumber\\
	= \max_{i\in[L]}&\left(-\dfrac{1}{k\mu_i}\log(1-\dfrac{k_i}{n_i})\right). \label{Eqn:LEMMELL}
	\end{align} 
\end{lemma}
\begin{proof}
	The proof is located at Appendix \ref{proof:Lemma:LEMMEL}.
\end{proof}

This lemma signals that the expected computing time of an $(\bm{n},\bm{k})-$group code can be easily obtained when $\vct{n},\vct{k}$ and $\vct{\mu}$ are given. 

\subsection{Optimizing Task Allocation}
In this subsection, we present the optimal task allocation rule $\vct{k}^*$ for given parameters $\vct{n}, \bm{\mu}$ and $k$. Before optimizing the task allocation vector $\vct{k}$, we provide a relationship between order statistics $T_{k:n}$ and $\{T^{(i)}_{k_i:n_i}\}_{i=1}^L$. 
%Now, we suggest bounds on the $k^{th}$ order statistic $T_{k:n}$ in terms of $L$ independent $k_i^{th}$ order statistics to compare $T_{\mathrm{comp\_MDS}}$ and $T_{\mathrm{comp\_G}}$ in the following proposition. 
\begin{lemma}\label{Lemma:Bounds_L}
	Under the scenario of computing $k$ tasks on an $(\bm{n},\bm{\mu})-$group system with $L$ groups, consider applying an $(\bm{n},\bm{k})-$group code where $\bm{n}=[n_1,n_2,\dots,n_L]$ and $\bm{k}=[k_1,k_2,\dots,k_L]$. Given an arbitrary realization of completion time $\{T_j^{(i)}\}_{i \in [L], j \in [n_i]}$ of workers, let $T_{k:n}$ be the $k^{th}$ smallest value among $\{T_j^{(i)}\}_{i\in[L],j\in[n_i]}$. Meanwhile, $T_{k_i:n_i}^{(i)}$ denotes the $k_i^{th}$ smallest value among $\{T_j^{(i)}\}_{j=1}^{n_i}$. Then, we have
	\begin{equation*}
	\min_{i\in[L]} T^{(i)}_{k_i:n_i} \leq T_{k:n} \leq \max_{i\in[L]} T^{(i)}_{k_i:n_i}.
	\end{equation*}
\end{lemma}
% It is natural that $\max(0,k-\sum_{j\neq i} k_j) < k_i < \min(n_i, k)$
\begin{proof}
	The proof can be found at Appendix \ref{proof:prop:boundL}
\end{proof}
Here we recall that $T_{k:n}$ is the computing time of an $(n,k)-$MDS code and the upper bound $\Max{i\in[L]} T^{(i)}_{k_i:n_i}$ is the computing time of an $(\bm{n},\bm{k})-$group code. Now, Theorem \ref{Theorem:ExecTime_L} specifies the optimal task allocation $\bm{k}^*$ defined as \eqref{Eqn:kStar} when there are $L$ groups. Moreover, the computing time of an $(\vct{n}, \vct{k}^*)-$group code and an $(n,k)-$MDS code is compared. 
\begin{theorem} \label{Theorem:ExecTime_L}
	Consider the scenario of computing $k$ tasks on an $(\bm{n},\bm{\mu})-$group system with $L$ groups, where an $(\bm{n},\bm{k})-$group code is applied.
	In the asymptotic regime of large $n$, the optimal task allocation $\bm{k}^*=[k_1^*,k_2^*, \cdots, k_L^*]$ can be obtained\footnote{Here we assume that $k_i^*$ is an integer for $i \in [L]$ since task allocation vector $\vct{k}$ consists of integer values. However, in the case of $k_i^*$ not an integer, we can use the round function to set $\vct{k}^*=[\nint{k_1^*}, \nint{k_2^*}, \cdots, \nint{k_L^*}]$. For reasonably large $n$ and $k$, this rounding function has a negligible impact on the overall performance.}
	by solving the following equations for $i\in[L]$:
	\begin{align}
	k_i^* + \sum_{j\neq i} n_j\left(1-\left(1-\frac{k_i^*}{n_i}\right)^{\frac{\mu_j}{\mu_i}}\right) = k.
	\end{align}
	Moreover, the corresponding expected computing time of an $(\bm{n},\bm{k}^*)-$group code is equal to that of an $(n,k)-$MDS code as $n$ goes to infinity, i.e.
	\begin{equation*}	
	\lim_{n\to\infty} \mathbb{E}[\TGO]=\lim_{n\to\infty} \mathbb{E} [\TM].
	\end{equation*} 
\end{theorem}
\begin{proof}
	We prove this theorem at Appendix \ref{proof:theorem:ExecTime_L}.
\end{proof}
% ..\cmt{mention the uniqueness and existence of $k^*$. Is it important information?}
Recall that an $(n,k)-$MDS code is optimal in terms of computing time. The above theorem illustrates that an $(\vct{n},\vct{k})-$group code can asymptotically achieve the optimal computing time when the tasks are optimally allocated, i.e. $\vct{k}=\vct{k}^*$.

\section{Decoding Time Analysis}\label{section:ENCDEC}
\begin{figure}[!t]
	\centering
	\includegraphics[width=60mm]{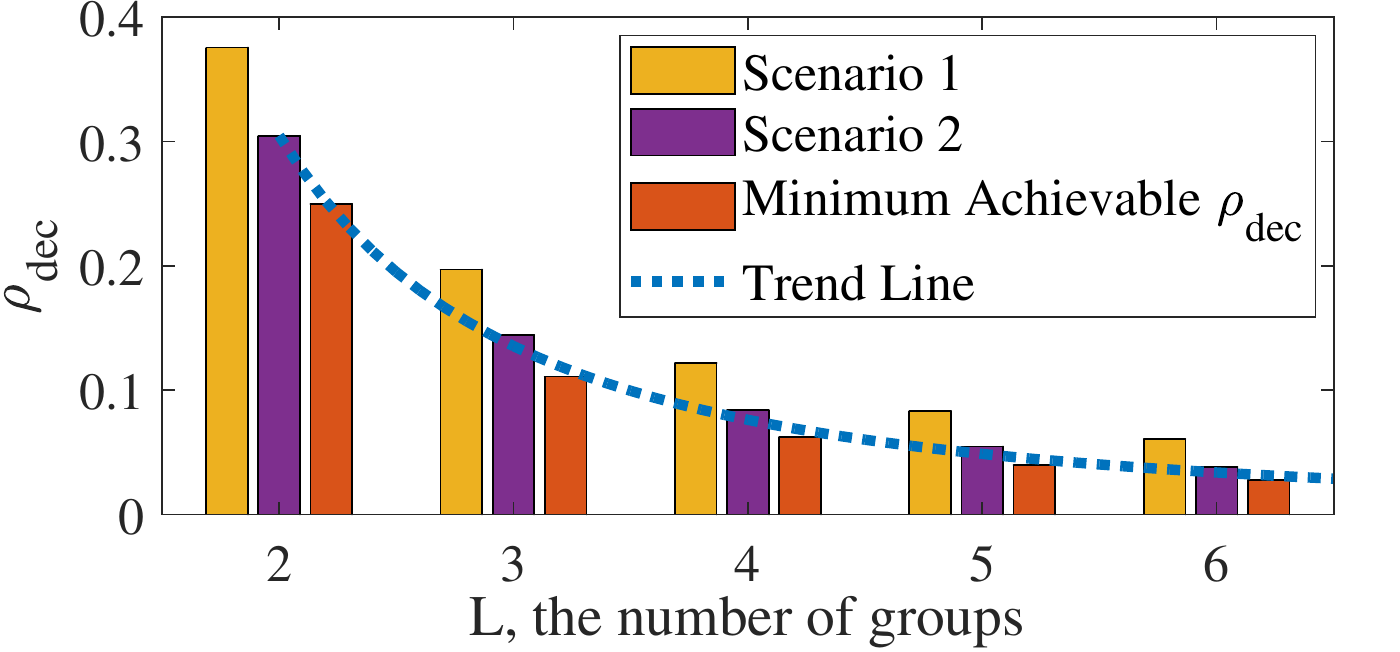}
	\caption{$\rho_\mathrm{dec}$ versus $L$ for three different scenarios; an imbalanced, a balanced and the best ones.}
	\label{Fig:dec_t}
\end{figure}

Now we compare the decoding complexity of the suggested $(\vct{n},\vct{k})-$group code to that of an $(n,k)-$MDS code.
%One can easily notice that it always takes less to decode the suggested $(\vct{n},\vct{k})-$group code compared to an $(n,k)-$MDS code, because $L$ independent $(n_i,k_i)-$MDS codes consisting the group code can be decoded separately and thereby in parallel. 
We assume that the decoding complexity of an $(n,k)-$MDS code is $\mathcal{O}(k^\beta)$ for $\beta> 1$\footnote{According to the recent works \cite{ref:balance1, ref:balance2} on decoding algorithms, practical scenarios satisfy $\beta>1$.}.
%Decoding with $\beta=1$ requires a large field size \cite{ref:polycode}.}. 
Then, the suggested $(\vct{n},\vct{k})-$group code has a decoding complexity of $\mathcal{O}((k_{\mathrm{max}})^\beta)$ by the virtue of parallel decoding, where $k_\mathrm{max}=\Max{i\in[L]} k_i$. Note that decoding complexities of two schemes grow with different orders. For a comparison, we define the ratio of the two orders as
\begin{equation*}
\rho_\mathrm{dec}=\left(\dfrac{k_\mathrm{max}}{k}\right)^{\beta}.
\end{equation*}
Note that the ratio $\rho_\mathrm{dec}$ can be minimized down to $(1/L)^\beta$ when we have $k_\mathrm{max}=k/L$. 

Fig. \ref{Fig:dec_t} illustrates $\rho_\mathrm{dec}$ under two different scenarios for given $n=240$ and $k=120$. In both scenarios, $\vct{n}$ and $\vct{\mu}$ 
are randomly generated. Moreover, the task allocations for both scenarios are selected as the optimal $\vct{k}^*$, depending on the given parameters of $\vct{n}$ and $\vct{\mu}$.
%Both scenarios consider when $\vct{n}$ and $\vct{\mu}$ are randomly generated.
Motivated by the practical setting where the size of each group and the average computing time of each worker are bounded, we set $\vct{n}\sim \mathrm{unif}(0.7\frac{n}{L},1.3\frac{n}{L})$ and $\vct{\mu}\sim \mathrm{unif}(1,2)$ with uniform distributions. 
Scenarios 1 and 2 differ in the rule of ordering the elements of $\vct{n}$ and $\vct{\mu}$, as illustrated below. For scenario 1, we sort the elements of $\vct{n}$ and $\vct{\mu}$ in ascending and descending order, respectively. In other words, $n_i \leq n_j$ and $\mu_i \geq \mu_j$ hold for all $i < j$. This is the scenario when a group with less average response time has less workers.
In the case of scenario 2, both $\vct{n}$ and $\vct{\mu}$ are sorted in ascending order, i.e., $n_i \geq n_j$ and $\mu_i \geq \mu_j$ hold for $i < j$. This is the scenario when a group with less average response time has more workers. 
Under these scenarios, we obtain the average values of $\rho_\mathrm{dec}$ for $10^4$ samples when $\beta=2$. The simulations on two scenarios are compared to the minimum achievable $\rho_\mathrm{dec}=(1/L)^{\beta}$. Moreover, we plotted the trend line, which is set to stretch from the point of Scenario 2 for $L=2$ and grow by a factor of $(1/L)^\beta$. 

Fig. \ref{Fig:dec_t} delineates that $\rho_\mathrm{dec}$ diminishes along with the trend line under any scenarios as $L$ grows. Combining this with the definition of $\rho_\mathrm{dec}$, we can remark that $k_\mathrm{max}$ is inversely proportional to $L$ in practical scenarios. Moreover, the proposed group code provides a significant decoding complexity reduction in both scenarios. For example, when $L=4$, an $(\vct{n},\vct{k})-$group code already achieves roughly 10x reduced decoding complexity compared to an $(n,k)-$MDS code. 

 \begin{table}[t!]
	\renewcommand{\arraystretch}{1.3}
	\centering
	% Some packages, such as MDW tools, offer better commands for making tables
	% than the plain LaTeX2e tabular which is used here.
	\footnotesize
	\caption{Code parameters and decoding complexities of various coding schemes used for the simulation.}
	\scalebox{0.9}{\begin{tabular}{|c||c|c|}
			\hline
			\multirow{2}{*}{Code}& Decoding  & Code  \\
			 &Complexity &  Parameters  \\
			\hline
			MDS  & $\mathcal{O}(k^\beta)$ & $(n,k)=(900,400)$  \\
			\hline
			Product  & $\mathcal{O}((\sqrt{k})^{\beta+1})$& $(\sqrt{n},\sqrt{k})^2=(30,20)^2$  \\
			\hline
			\multirow{2}{*}{Group }  &\multirow{2}{*}{$\mathcal{O}(k_\mathrm{max}^\beta)$}& $\vct{n}=[180,170,160,140,130,120]$\\
			& & $\vct{k}=\vct{k}^*=[71,71,70,65,63,60]$\\
			\hline     
		\end{tabular}\label{tab:code_info}
	}
\end{table}

Now we compare the total execution time of the suggested group code to existing schemes by using a simulation. We represent the total execution time as $T_\mathrm{exec}= T_\mathrm{comp}+\alpha T_\mathrm{dec}$, where the coefficient $\alpha\ge 0$ indicates a relative weight of the decoding complexity compared to the computing time. We simulate the computing of $k=400$ tasks on an $(\vct{n},\vct{\mu})-$group system with $\vct{n}=[180,170,160,140,130,120]$ and $\vct{\mu}=[1.25, 1.35,1.45,1.55,1.65,1.75]$, which leads to $n=900$ with $L=6$ groups. For varying $\alpha$, we observe the execution times of the MDS code, the product code, and the suggested group code with parameters listed on Table \ref{tab:code_info}. The decoding complexity of the product code is $\mathcal{O}((\sqrt{k})^{\beta+1})$ because the decoding procedure consists of decoding $2\sqrt{k}$ MDS codes, where the dimension of each MDS code is $\sqrt{k}$. For the group code, we use the optimal task allocation rule $\vct{k}=\vct{k}^*$. For the decoding complexity, we use a parameter of $\beta=2$.

Fig. \ref{Fig:simul_low_alpha} and Fig. \ref{Fig:simul_large_alpha} show the simulated execution times for different regimes of $\alpha$. Fig. \ref{Fig:simul_low_alpha} illustrates the situation where the computing time is dominant, i.e. $\alpha$ is small. When $\alpha$ is the lowest in Fig. \ref{Fig:simul_low_alpha}, the MDS code gives the smallest execution time, followed by the group code and then the product code. This 
coincides with the two mathematical results shown above: the optimality of the MDS code in Lemma \ref{Lemma:OrdBehavior} and the asymptotic optimality of the group code in Theorem \ref{Theorem:ExecTime_L}. Note that the coding scheme that gives the best execution time changes as $\alpha$ varies. Meanwhile, Fig. \ref{Fig:simul_large_alpha} represents the situation where the decoding complexity dominates the execution time. Notice that the execution time of the MDS code becomes inferior to other schemes due to its huge decoding complexity as $\alpha$ grows. On this computing system, the group code gives the best execution time for all regime of $\alpha$. In general, the order of $k_\mathrm{max}$ determines which of the group code or the product code has a better decoding complexity. Recall that the decoding complexity of the group code and the product code are $\mathcal{O}(k_\mathrm{max}^\beta)$ and $\mathcal{O}((\sqrt{k})^{\beta+1})$, respectively. Thus, we can say that the decoding complexity of the group code is better than the product code when 
\begin{equation}
k_\mathrm{max}=\mathcal{O}((\sqrt{k})^{1+\frac{1}{\beta}})\label{eqn:condition}
\end{equation} holds.
Remind that $k_\mathrm{max}$ is inversely proportional to $L$ under practical scenarios as shown in Fig. \ref{Fig:dec_t}. Thus, the condition in \eqref{eqn:condition} reduces to  $L=\Omega\left(\dfrac{1}{(\sqrt{k})^{1+\frac{1}{\beta}}}\right)$.
This implies that when a system has sufficiently large number of groups, the group code outperforms the product code in terms of the decoding complexity. 

\begin{figure}[!t]
	\centerline{\subfigure[Low-$\alpha$ regime]{\includegraphics[height=0.41\columnwidth]{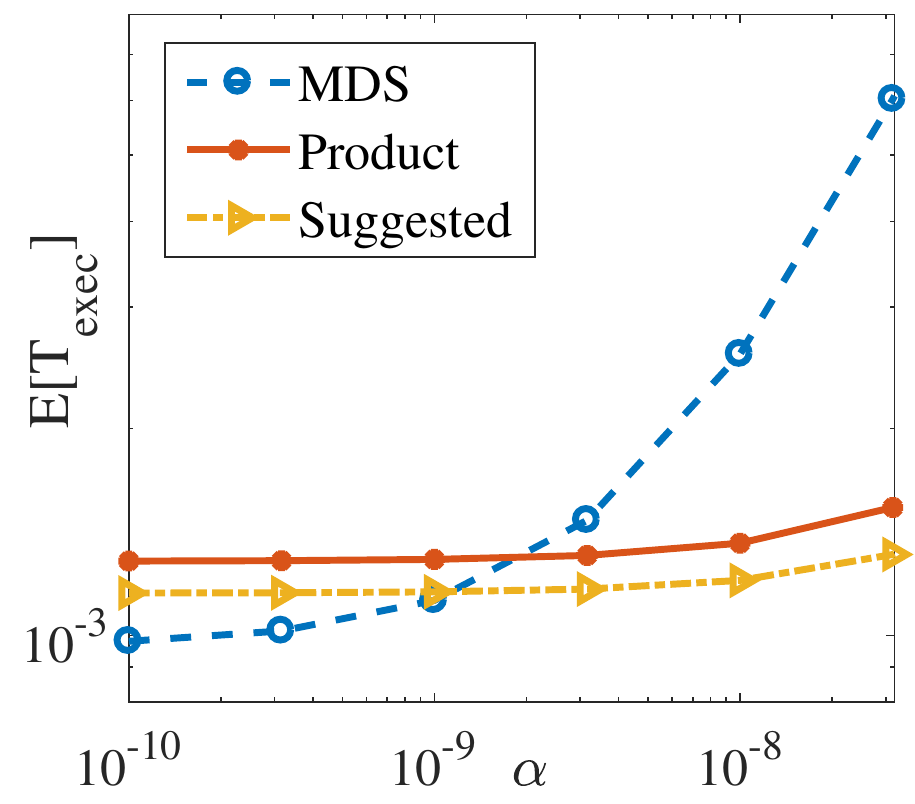}
			% where an .eps filename suffix will be assumed under latex,
			% and a .pdf suffix will be assumed for pdflatex
			\label{Fig:simul_low_alpha}
		}
		\hfil
		\subfigure[Large-$\alpha$ regime]{\includegraphics[height=0.41\columnwidth]{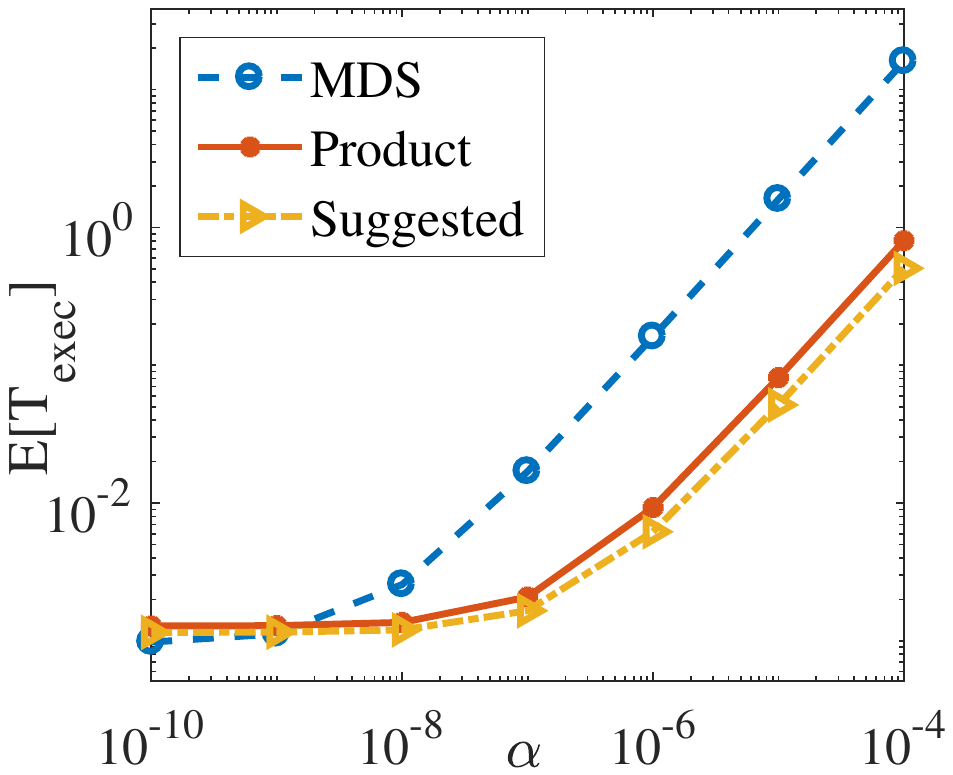}
			% where an .eps filename suffix will be assumed under latex,
			% and a .pdf suffix will be assumed for pdflatex
			\label{Fig:simul_large_alpha}
	}}
	\caption{Simulated results of $\mathbb{E}[T_\text{exec}]$ under various coding schemes.}
	\label{Fig:simul_compare}
\end{figure}

\section{Conclusion}\label{section:conclusion}
In this paper, we propose a coded computation scheme appropriate for a practical model, which reflects the tree-shaped structure and the heterogeneity of groups. Precisely, we consider systems with $L$ heterogeneous groups that have distinct computing time statistics and a different number of workers. We prove that the suggested group-coded scheme can asymptotically achieve the optimal computing time as $n$ grows to infinity. In the regime of finite $n$, numerical results show that the suggested scheme also provides a near-optimal computing time. 
Moreover, the suggested scheme can reduce the decoding complexity down to a factor of $(\frac{1}{L})^\beta$, where $\beta>1$, compared to the existing MDS coded scheme. Finally, the total execution time$-$the sum of the computing time and the decoding time$-$of the suggested scheme is numerically shown to outperform other existing state-of-the-art coding schemes.

\appendices
\section{Proof of Lemma \ref{Lemma:OrdBehavior}}\label{proof:lemmel2}
We first show that $\max(T_{k_1:n_1}^{(1)}, T_{k_2:n_2}^{(2)})$ is determined as one of $T_{k_1:n_1}^{(1)}$ and $T_{k_2:n_2}^{(2)}$ for sufficiently large $n$, and thereby the expected value of $\max(T_{k_1:n_1}^{(1)}, T_{k_2:n_2}^{(2)})$ is determined as the maximum among the expected values of $T_{k_1:n_1}^{(1)}$ and $T_{k_2:n_2}^{(2)}$.

First, consider the $k^{th}$ order statistic of i.i.d. $n$ random variables $T_{k:n}$, whose probability distribution function (PDF) and cumulative distribution function (CDF) are denoted by $f(\cdot)$ and $F(\cdot)$. We represent an empirical CDF obtained with $n$ samples as $\tilde{F}_n(\cdot)$. According to \cite{ref:aymptotic}, $T_{k:n}$ can be represented as 
\begin{equation}\label{eqn:CvgcOrdStat}
T_{k:n} = \xi - \frac{\widetilde{F}_n(\xi)-k/n}{f(\xi)} +R_n,
\end{equation}
where $\xi=F^{-1}(k/n)$ and the third term $R_n$ satisfies $n^{1/2}R_n \xrightarrow{p} 0.$ %Here, the asymptotic distribution of $n^{\frac{1}{2}}(T_{k:n}-\xi)$ is a normal distribution with zero mean and variance $\dfrac{(k/n)(1-k/n)}{f^2(\xi)}$, i.e., 
In \cite{ref:aymptotic}, it is shown that $n^{1/2}(T_{k:n}-\xi) \xrightarrow{d} X$, where $X \sim N \left(0, \dfrac{(k/n)(1-k/n)}{f^2(\xi)}\right)$. Thus, we have
$ T_{k:n} \xrightarrow{d} N \left(\xi,\dfrac{(k/n)(1-k/n)}{nf^2(\xi)}\right).$ 

Now, we examine the convergence of $T_{k_1:n_1}^{(1)} - T_{k_2:n_2}^{(2)}$ by using \eqref{eqn:CvgcOrdStat}.
% We define $r_1, r_2 \in (0,1)$ as constants such that $k_1=r_1n_1$ and $k_2=r_2n_2$. 
Let $f_{(i)}$ and $F_{(i)}$ be the PDF and CDF of an exponential random variable with rate $k\mu_i$, and define $\xi^{(i)}$ as $\xi^{(i)} = F^{-1}_{(i)}(k_i/n_i)$ for $i=1,2$, i.e.
\begin{align}
\xi^{(i)} = F^{-1}_{(i)}(k_i/n_i)=-\frac{1}{k\mu_i}\log\left(1-\frac{k_i}{n_i}\right) \label{Eqn:exp_mean}
\end{align}
Then, we can think of the asymptotic distribution of the $T_{k_1:n_1}^{(1)} - T_{k_2:n_2}^{(2)}$ as follows:
\begin{align*}
&T_{k_1:n_1}^{(1)} - T_{k_2:n_2}^{(2)} -  (\xi^{(1)} - \xi^{(2)})\\
= & \hspace{0.1cm} n^{-\frac{1}{2}} n^{\frac{1}{2}} [(T_{k_1:n_1}^{(1)} - \xi^{(1)}) - (T_{k_2:n_2}^{(2)} - \xi^{(2)})]
\xrightarrow{d}  \hspace{0.1cm} n^{-\frac{1}{2}} Z_V,
\end{align*}
where $Z_V \sim N(0,V)$ for $V = \dfrac{\frac{k_1}{n_1}(1-\frac{k_1}{n_1})}{f_{(1)}^2(\xi^{(1)})}+\dfrac{\frac{k_2}{n_2}(1-\frac{k_2}{n_2})}{f^2_{(2)}(\xi^{(2)})}$. 
%Then, the term, $n^{-\frac{1}{2}}Z_V$, follows $N(0,V/n)$ since it is a constant multiple of a normal random variable. 
By the definition of convergence in distribution,
% the CDF of $T_{k_1:n_1}^{(1)} - T_{k_2:n_2}^{(2)} - (\xi^{(1)} - \xi^{(2)})$ converges into that of $n^{-\frac{1}{2}}Z_V$, i.e., 
for any $\epsilon>0$, we have
\begin{equation*}
\lim_{n\to\infty} \Pr\{T_{k_1:n_1}^{(1)} - T_{k_2:n_2}^{(2)} - (\xi^{(1)} - \xi^{(2)})\leq \epsilon\}\nonumber
= \lim_{n\to\infty} \Phi\left(\epsilon\sqrt{\frac{n}{V}}\right).
\end{equation*}
Then, the convergence of  $T_{k_1:n_1}^{(1)} - T_{k_2:n_2}^{(2)}$ into $\xi^{(1)}-\xi^{(2)}$ can be derived as follows.
\begin{align}
&\lim_{n\to\infty} \Pr \left(\abs{T_{k_1:n_1}^{(1)} - T_{k_2:n_2}^{(2)} - (\xi^{(1)} - \xi^{(2)})} \geq \epsilon \right)\nonumber\\
&=\lim_{n\to\infty} 2\left(1-\Phi\left(\epsilon\sqrt{\frac{n}{V}}\right)\right)=\hspace{0.1cm} 0. \label{Eqn:convinprob}
\end{align}
%&\leq \lim_{n\to\infty} \frac{\sqrt{2V}}{\epsilon\sqrt{n\pi}}\cdot \exp^{-n\epsilon^2/2V} 
This means $T_{k_1:n_1}^{(1)} - T_{k_2:n_2}^{(2)}$ converges in probability towards the constant $\xi^{{(1)}} - \xi^{(2)}$ as $n\to\infty$, i.e.
\begin{equation*}
T_{k_1:n_1}^{(1)} - T_{k_2:n_2}^{(2)} \xrightarrow{p}  \xi^{{(1)}} - \xi^{(2)}.
\end{equation*}
It illustrates that for sufficiently large $n$, the order of two independent order statistics is maintained corresponding to their mean values due to the convergence. 
%This holds for any value of $\xi^{(1)}-\xi^{(2)}$. When $\xi^{(1)}-\xi^{(2)}>0$, the following equation holds by using \eqref{Eqn:convinprob} with $\epsilon = \xi^{(1)}-\xi^{(2)}$.
%\begin{align*}
%&\lim_{n\to\infty}\Pr\left\{T_{k_1:n_1}^{(1)}-T_{k_2:n_2}^{(2)} < 0\right\}\\
%= &\lim_{n\to\infty}\Pr\left\{T_{k_1:n_1}^{(1)}-T_{k_2:n_2}^{(2)}-(\xi^{(1)}-\xi^{(2)})< -(\xi^{(1)}-\xi^{(2)}) \right\}\\
%\leq& \lim_{n\to\infty} \Pr\left\{|T_{k_1:n_1}^{(1)}-T_{k_2:n_2}^{(2)}-(\xi^{(1)}-\xi^{(2)})| \geq \xi^{(1)}-\xi^{(2)} \right\}
%\\=& \ \ 0 
%\end{align*} 
%In a similar procedure, $\lim_{n\to\infty}\Pr\left\{T_{k_1:n_1}^{(1)}-T_{k_2:n_2}^{(2)} > 0\right\}=0$ holds  when  $\xi^{(1)}-\xi^{(2)} < 0$ by inserting $\epsilon=\xi^{(2)}-\xi^{(1)}$ into \eqref{Eqn:convinprob}. Finally, when $\xi^{(1)}-\xi^{(2)} = 0$, the equation \eqref{Eqn:convinprob} becomes
%\begin{align*}
%&\lim_{n\to\infty}  \Pr \left\{|T_{k_1:n_1}^{(1)}-T_{k_2:n_2}^{(2)}|\geq \epsilon \right\} \\
%= &\lim_{n\to\infty} \Pr \left\{T_{k_1:n_1}^{(1)} \neq T_{k_2:n_2}^{(2)} \right\} =\hspace{0.2cm}0.
%\end{align*}
%This means the probability that $T_{k_1:n_1}^{(1)}$ and $T_{k_2:n_2}^{(2)}$ have different values goes to zero as $n\to\infty$. 
Consequently,
% For any value of  $\xi^{(1)}-\xi^{(2)}$, 
the sign of $T_{k_1:n_1}^{(1)}-T_{k_2:n_2}^{(2)}$ loses randomness and is determined in asymptotic regime of large $n$. Therefore, we can claim that
\begin{equation}\label{Eqn:orderofT}
\lim_{n\to\infty}\mathbbm{1}_{T_{k_1:n_1}^{(1)}>T_{k_2:n_2}^{(2)}} = \mathbbm{1}_{\xi^{(1)}>\xi^{(2)}}.
\end{equation}
This equation indicates that in asymptotic regime of large $n$, the random variable $\mathbbm{1}_{T_{k_1:n_1}^{(1)}>T_{k_2:n_2}^{(2)}}$, which has cumbersome distribution, can be substituted with $\mathbbm{1}_{\xi^{(1)}>\xi^{(2)}}$, which is a binary number that can be easily calculated.
%This equality is a key to obtain the asymptotic expectation the computing time of the group code.

Now we prove the statement of Lemma \ref{Lemma:OrdBehavior} by using \eqref{Eqn:orderofT} as follows.
\begin{align}	
&\lim_{n\to\infty}\mathbb{E}\left[\max\left(T_{k_1:n_1}^{(1)},T_{k_2:n_2}^{(2)}\right)\right] \nonumber\\
\overset{\mathrm{(a)}}{=}& \hspace{0.05cm}\mathbb{E}\left[\lim_{n\to\infty}\max\left(T_{k_1:n_1}^{(1)},T_{k_2:n_2}^{(2)}\right)\right] \nonumber\\
=&\hspace{0.1cm}\mathbb{E} \bigg[ \lim_{n\to\infty}{T_{k_1:n_1}^{(1)}}\cdot\lim_{n\to\infty}{\mathbbm{1}_{T_{k_1:n_1}^{(1)} \geq T_{k_2:n_2}^{(2)}}} \nonumber\\
& \hspace{2cm} +\lim_{n\to\infty}{T_{k_2:n_2}^{(2)}}\cdot\lim_{n\to\infty} \mathbbm{1}_{T_{k_1:n_1}^{(1)} < T_{k_2:n_2}^{(2)}} \bigg] \nonumber\\
\overset{\mathrm{(b)}}{=}& \mathbb{E}\left[\lim_{n\to\infty}{T_{k_1:n_1}^{(1)}}\cdot \mathbbm{1}_{\xi^{(1)} \geq \xi^{(2)}} 
+\lim_{n\to\infty}{T_{k_2:n_2}^{(2)}}\cdot{\mathbbm{1}_{\xi^{(1)} < \xi^{(2)}}}\right]\nonumber\\
=& \mathbb{E}\left[\lim_{n\to\infty}{T_{k_1:n_1}^{(1)}}\right]\cdot \mathbbm{1}_{\xi^{(1)} \geq \xi^{(2)}}
+\mathbb{E}\left[\lim_{n\to\infty}{T_{k_2:n_2}^{(2)}}\right]\cdot{\mathbbm{1}_{\xi^{(1)} < \xi^{(2)}}}\nonumber\\
=& \mathbb{E}\left[\lim_{n\to\infty} T_{k_1:n_1}^{(1)}\right]\cdot \mathbbm{1}_{\mathbb{E}\left[\Lim{n\to\infty}{T_{k_1:n_1}^{(1)}}\right] \geq \mathbb{E}\left[\Lim{n\to\infty}{T_{k_2:n_2}^{(2)}}\right]}\nonumber\nonumber\\
&\hspace{0.5mm}+\mathbb{E}\left[\lim_{n\to\infty} T_{k_2:n_2}^{(2)}\right] \cdot \mathbbm{1}_{ \mathbb{E}\left[\Lim{n\to\infty}{T_{k_1:n_1}^{(1)}}\right] < \mathbb{E}\left[\Lim{n\to\infty}{T_{k_2:n_2}^{(2)}}\right]}\nonumber\\
=& \max\left(\mathbb{E}\left[\lim_{n\to\infty}T_{k_1:n_1}^{(1)}\right],\mathbb{E}\left[\lim_{n\to\infty}T_{k_2:n_2}^{(2)}\right]\right)\nonumber\\
=& \max\left(\lim_{n\to\infty}\mathbb{E}\left[T_{k_1:n_1}^{(1)}\right],\lim_{n\to\infty}\mathbb{E}\left[T_{k_2:n_2}^{(2)}\right]\right)\nonumber \\
=&\max\left(-\dfrac{1}{k\mu_1}\log(1-\dfrac{k_1}{n_1}),-\dfrac{1}{k\mu_2}\log(1-\dfrac{k_2}{n_2})\right)\nonumber
\end{align}
Equality $(a)$ holds since limit and expectation can be interchanged when the random variable is non-negative, which is satisfied because $\max(T_{k_1:n_1}^{(1)},T_{k_2:n_2}^{(2)})\geq0$. Equality $(b)$ holds by \eqref{Eqn:orderofT}. Note that this proof can be directly applied to the min function of two independent order statistics instead of max function.

\section{Proof of Lemma \ref{Lemma:LEMMELL}}\label{proof:Lemma:LEMMEL}
We prove the statement by using the mathematical induction. For the base step, we already prove the statement for $L=2$ in Lemma \ref{Lemma:OrdBehavior}. Now, we show if the statement is true for an arbitrary $L>2$, then the statement still holds for $L+1$. Before moving onto the proof, we provide the convergence of max function, which is necessary for the proof. Recall that equation \eqref{Eqn:orderofT} shows the order of two independent order statistics is determined by their expectation values for sufficiently large $n$. Thus, we can claim for arbitrary $\gamma,\delta \in [L]$, the following statement is true.
\begin{equation*}
\lim_{n\to\infty}\mathbbm{1}_{T_{k_{\gamma}:n_{\gamma}}^{(\gamma)}>T_{k_{\delta}:n_{\delta}}^{(\delta)}} = \mathbbm{1}_{\xi^{(\gamma)}>\xi^{(\delta)}}
\end{equation*}
This leads to 
\begin{equation}
\lim_{n\to\infty} \max_{i\in[L]} T^{(i)}_{k_i:n_i} = T^{(i_{\mathrm{max}})}_{k_{i_{\mathrm{max}}}:n_{i_{\mathrm{max}}}},\label{eqn:MaxofL}
\end{equation}
where $i_{\mathrm{max}}=\underset{i\in[L]}{\arg\max}\hspace{1mm}\xi^{(i)}.$ In other words, the maximum of $L$ independent order statistics is determined as the one that has the largest expectation value for sufficiently large $n$.

We here move on to the inductive step, assuming the statement holds for $L=L'$ as
\begin{equation}\label{eqn:LEMtoMEL2}
\lim_{n\to\infty} \mathbb{E}[\max_{i\in[L']} T^{(i)}_{k_i:n_i}] =\max_{i\in[L']} (\mathbb{E}[\lim_{n\to\infty} T_{k_i:n_i}^{(i)}]).
\end{equation}
Now, we examine the statement holds for $L'+1$ as well:
\begin{align*}
&\lim_{n\to\infty} \mathbb{E}[\max_{i\in[L'+1]} T^{(i)}_{k_i:n_i}]\\
&=\lim_{n\to\infty} \mathbb{E}[\max(  \max_{i\in[L']} T^{(i)}_{k_i:n_i}, T^{(L'+1)}_{k_{L'+1}:n_{L'+1}})]\\
&=\mathbb{E}[\lim_{n\to\infty} \max( \max_{i\in[L']} T^{(i)}_{k_i:n_i}, T^{(L'+1)}_{k_{L'+1}:n_{L'+1}})]\\
&\overset{\mathrm{(c)}}{=}\mathbb{E}[\lim_{n\to\infty} \max(T^{(i_{\mathrm{max}})}_{k_{i_{\mathrm{max}}}:n_{i_{\mathrm{max}}}}, T^{(L'+1)}_{k_{L'+1}:n_{L'+1}})]\\
&\overset{\mathrm{(d)}}{=}\max( \mathbb{E}[\lim_{n\to\infty} T^{(i_{\mathrm{max}})}_{k_{i_{\mathrm{max}}}:n_{i_{\mathrm{max}}}}], \mathbb{E}[\lim_{n\to\infty} T^{(L'+1)}_{k_{L'+1}:n_{L'+1}}])\\
&\overset{\mathrm{(e)}}{=}\max( \mathbb{E}[\lim_{n\to\infty} \max_{i\in[L']} T^{(i)}_{k_i:n_i}], \mathbb{E}[\lim_{n\to\infty} T^{(L'+1)}_{k_{L'+1}:n_{L'+1}}])\\
&=\max( \lim_{n\to\infty} \mathbb{E}[ \max_{i\in[L']} T^{(i)}_{k_i:n_i}], \mathbb{E}[\lim_{n\to\infty} T^{(L'+1)}_{k_{L'+1}:n_{L'+1}}])\\
&=\max( \max_{i\in[L']} (\mathbb{E}[\lim_{n\to\infty} T_{k_i:n_i}^{(i)}]), \mathbb{E}[\lim_{n\to\infty} T^{(L'+1)}_{k_{L'+1}:n_{L'+1}}])\\
&= \max_{i\in[L'+1]} (\mathbb{E}[\lim_{n\to\infty} T_{k_i:n_i}^{(i)}])\\
&= \max_{i\in[L'+1]} (\lim_{n\to\infty}\mathbb{E}[ T_{k_i:n_i}^{(i)}]).
\end{align*}
Equality $(\mathrm{c})$ holds since $ \underset{i\in[L']}{\max}\hspace{1mm}T^{(i)}_{k_i:n_i}$ becomes the one whose expectation value is the largest for sufficiently large $n$ as shown in \eqref{eqn:MaxofL}. We can lead to equality $(\mathrm{d})$ by Lemma \ref{Lemma:OrdBehavior} since it is equivalent to the case when $L=2$. Equality $(\mathrm{e})$ holds by the assumption \eqref{eqn:LEMtoMEL2}. Thus, we have 
\begin{equation*}
\lim_{n\to\infty} \mathbb{E}[\max_{i\in[L'+1]} T^{(i)}_{k_i:n_i}] =\max_{i\in[L'+1]} (\lim_{n\to\infty}\mathbb{E}[ T_{k_i:n_i}^{(i)}]),
\end{equation*}
which completes the whole proof of this lemma. Similarly, we can show 
\begin{equation*}
\lim_{n\to\infty} \mathbb{E}[\min_{i\in[L']} T^{(i)}_{k_i:n_i}] =\min_{i\in[L']} (\lim_{n\to\infty} \mathbb{E}[T_{k_i:n_i}^{(i)}]).	\qedhere
\end{equation*}

\section{Proof of Lemma \ref{Lemma:Bounds_L}}\label{proof:prop:boundL}
Imagine there are three groups. Then, for arbitrary realization of $\{T_j^{(i)}\}_{i \in [3], j \in [n_i]}$, the following inequalities hold by Lemma \ref{proposition:Bounds_L2}.
\begin{align*}
\min\left(T_{k_1:n_1}^{(1)},\max( T^{(2)}_{k_2:n_2}, T^{(3)}_{k_3:n_3})\right) &\leq T_{k:n} \\
\leq \max \Big(T_{k_1:n_1}^{(1)},&\max( T^{(2)}_{k_2:n_2}, T^{(3)}_{k_3:n_3}) \Big).
\end{align*}
%This is because the time taken for gathering $k-k_1=k_2+k_3$ responses from group $2$ and $3$ is maximum of $T^{(2)}_{k_2:n_2}$ and $ T^{(3)}_{k_3:n_3}$. 
We can change the lower bound by using an apparent inequality $\min( T^{(2)}_{k_2:n_2}, T^{(3)}_{k_3:n_3})\leq \max( T^{(2)}_{k_2:n_2}, T^{(3)}_{k_3:n_3})$ to have
\begin{align*}
\min\left(T_{k_1:n_1}^{(1)},\min( T^{(2)}_{k_2:n_2}, T^{(3)}_{k_3:n_3})\right) &\leq T_{k:n} \\
\leq \max \Big(T_{k_1:n_1}^{(1)},&\max( T^{(2)}_{k_2:n_2}, T^{(3)}_{k_3:n_3}) \Big).
\end{align*} 
Thus, we have
\begin{align*}
\min(T_{k_1:n_1}^{(1)}, T^{(2)}_{k_2:n_2}, T^{(3)}_{k_3:n_3})&\leq T_{k:n}\\ 
\leq& \max(T_{k_1:n_1}^{(1)}, T^{(2)}_{k_2:n_2}, T^{(3)}_{k_3:n_3}).
\end{align*}
We can also prove the statement for an arbitrary $L\geq2$ by repeating this process. Thus, we have
\begin{equation*}
\min_{i\in[L]} T^{(i)}_{k_i:n_i} \leq T_{k:n} \leq  \max_{i\in[L]} T^{(i)}_{k_i:n_i}.\qedhere
\end{equation*}

\section{Proof of Theorem \ref{Theorem:ExecTime_L}}\label{proof:theorem:ExecTime_L}
We first prove that the best task allocation rule $\vct{k}^*$ satisfies that the following equations:
\begin{align}
\lim_{n\to\infty}\mathbb{E}[ T_{k_i^*:n_i}^{(i)}]=\lim_{n\to\infty}\mathbb{E}[ T_{k_j^*:n_j}^{(j)}] \hspace{0.2cm} \mathrm{for} \hspace{0.2cm} i, j\in[L]. \label{Eqn:same_exp_t}
\end{align}
Then, we show the an $(\vct{n},\vct{k}^*)-$group code achieves the same computing time as an $(n,k)-$MDS code in an asymptotic region of large $n$. Afterwards, we provide the proof of the existence and the uniqueness of $\vct{k}^*$.

First, we rewrite the the statement \eqref{Eqn:LEMMELL} of Lemma \ref{Lemma:LEMMELL} w.r.t. $\Lim{n\to\infty}\mathbb{E}[ T_{k_j:n_j}^{(j)}]$ for $j\in[L]$ as follows:
\begin{align*}
\lim_{n\to\infty}& \mathbb{E}[\TG] =\max_{i\in[L]} (\lim_{n\to\infty}\mathbb{E}[ T_{k_i:n_i}^{(i)}]) \\
&= \max(\lim_{n\to\infty}\mathbb{E}[ T_{k_j:n_j}^{(j)}], \max_{i\neq j} (\lim_{n\to\infty}\mathbb{E}[ T_{k_i:n_i}^{(i)}]))\\
&=\max(-\dfrac{1}{k\mu_j}\log(1-\dfrac{k_j}{n_j}), \max_{i\neq j} (\lim_{n\to\infty}\mathbb{E}[ T_{k_i:n_i}^{(i)}])).
\end{align*}

Note that the first variable of the max function is a strictly increasing convex function with $k_j$, whereas the second variable $\Max{i\neq j} (\Lim{n\to\infty}\mathbb{E}[ T_{k_i:n_i}^{(i)}])$ is a strictly deceasing convex function with $k_j$ because it is equivalent to the time for computing $k-k_j$ tasks by using a group code with $L-1$ groups by \eqref{Eqn:LEMMELL}. Hence, taking max of the two variables results in a convex function that has the minimum value at the intersection of the two variables. Hence, the optimal value of $k_j= k_j^*$ satisfies
\begin{equation*}
\max_{i\neq j} (\lim_{n\to\infty}\mathbb{E}[ T_{k_i:n_i}^{(i)}])=\lim_{n\to\infty}\mathbb{E}[ T_{k_j^*:n_j}^{(j)}].
\end{equation*} 
We may write as below:
\begin{equation*}
\lim_{n\to\infty}\mathbb{E}[ T_{k_i:n_i}^{(i)}] \leq \lim_{n\to\infty}\mathbb{E}[ T_{k_j^*:n_j}^{(j)}].
\end{equation*} 
To satisfy the above inequality for all $i \neq j$ and $j\in[L]$, the optimal tast allocation $\vct{k}^*$ must satisfy the equation \eqref{Eqn:same_exp_t}.
%Note that the right-hand side has the exactly same form as the max function we start, thus equation \eqref{Eqn:same_exp_t} can be obtained by repeating this process for $L-1$ times.

Next, we consider the following bounds, which obtained by taking $\Lim{n\to\infty}\mathbb{E}[\cdot]$ of the bounds suggested in Lemma \ref{Lemma:Bounds_L} and applying Lemma \ref{Lemma:LEMMELL}:
\begin{align*}
\min_{i\in[L]} (\lim_{n\to\infty}\mathbb{E}[ T_{k_i:n_i}^{(i)}]) \leq \lim_{n\to\infty} \mathbb{E} [T_{k:n}] \leq 	\max_{i\in[L]} (\lim_{n\to\infty}\mathbb{E}[ T_{k_i:n_i}^{(i)}]).	
\end{align*}

For $\vct{k}=\vct{k}^*$, the above lower and upper bounds have an equal value by \eqref{Eqn:same_exp_t}. Hence, $\Lim{n\to\infty} \mathbb{E} [T_{k:n}] $ and $\Max{i\in[L]} (\Lim{n\to\infty}\mathbb{E}[ T_{k_i^*:n_i}^{(i)}])$ have the same value, which correspond to $\Lim{n\to\infty}\mathbb{E}[\TM]$ and $\Lim{n\to\infty}\mathbb{E}[\TGO]$ respectively. Thus, we prove
\begin{equation*}
\lim_{n\to\infty} \mathbb{E}[\TGO]=\lim_{n\to\infty} \mathbb{E} [\TM].
\end{equation*}

Lastly, we move on to the proof of the existence and the uniqueness of $\vct{k}^*$. Remark that the interval of $k_i^*$ is confined as $k_i^*\in[\max(0,k-n+n_i),\min(n_i,k)]$ due to the conditions $k_i\leq n_i$ and $k\leq n$. By inserting equation \eqref{Eqn:exp_mean} to \eqref{Eqn:same_exp_t}, the following equation is obtained for $i,j\in[L]$:
\begin{equation*}
k_j^*=n_j\left(1-\left(1-\frac{k_i^*}{n_i}\right)^{\frac{\mu_j}{\mu_i}}\right)
\end{equation*}
Thus, we may write the following equation which consists of a single variable $k_i^*$.
\begin{align*}
k=\sum_{i\in[L]} k_i^* = k_i^*+\sum_{j\neq i} n_j\left(1-\left(1-\frac{k_i^*}{n_i}\right)^{\frac{\mu_j}{\mu_i}}\right).
\end{align*}
For simplicity, we denote the right-hand side by $h(k_i^*)$. Note that $h(k_i^*)$ is a strictly increasing function with $k_i^*$. thus we can complete the proof if we show $h(k_i^*)$ starts from a value lower than $k$ and reaches to another value greater than $k$ in the given interval. Firstly, when the lower bound $\max(0,k-n+n_i)$ is $0$, it is obvious that $h(0)=0$. The other case, when $k-n+n_i>0$, is also easily proved as,
\begin{align*}
&h(k-n+n_i) \\
=& k-n+n_i  + \sum_{j\neq i} n_j\left(1-\left(1-\frac{k-n+n_i}{n_i}\right)^{\frac{\mu_j}{\mu_i}}\right)\\
=&k - \sum_{j\neq i} n_j\left(\frac{n-k}{n_i}\right)^{\frac{\mu_j}{\mu_i}} < k.
\end{align*}
Similarly, when the upper bound $\min(n_i,k)$ is $n_i$, one can easily show that $h(n_i)=n>k$. The other case of $\min(n_i,k)=k$, i.e. $k < n_i$, also satisfies $h(k) > k$ as follows.
\begin{align*}
h(k)& = k  + \sum_{j\neq i} n_j\left(1-\left(1-\frac{k}{n_i}\right)^{\frac{\mu_j}{\mu_i}}\right)\\
&=k+ \sum_{j\neq i} n_j\left(1-\left(\frac{n_i-k}{n_i}\right)^{\frac{\mu_j}{\mu_i}}\right) > k.
\end{align*}

We complete the proof by showing that $h(k_i^*)<k$ for the lower bound $k_i^*=\max(0,k-n+n_i)$ and $h(k_i^*)>k$ for the upper bound $k_i^*=\min(n_i,k)$, which guarantees the existence of the one intersection between a strictly increasing function $h(k_i^*)$ and a constant function $k$ with $k_i^*$.

\ifCLASSOPTIONcaptionsoff
\newpage
\fi

\bibliographystyle{IEEEtran}
\bibliography{IEEEabrv,ISIT19Kim}

\begin{thebibliography}{10}
\providecommand{\url}[1]{#1}
\csname url@samestyle\endcsname
\providecommand{\newblock}{\relax}
\providecommand{\bibinfo}[2]{#2}
\providecommand{\BIBentrySTDinterwordspacing}{\spaceskip=0pt\relax}
\providecommand{\BIBentryALTinterwordstretchfactor}{4}
\providecommand{\BIBentryALTinterwordspacing}{\spaceskip=\fontdimen2\font plus
\BIBentryALTinterwordstretchfactor\fontdimen3\font minus
  \fontdimen4\font\relax}
\providecommand{\BIBforeignlanguage}[2]{{%
\expandafter\ifx\csname l@#1\endcsname\relax
\typeout{** WARNING: IEEEtran.bst: No hyphenation pattern has been}%
\typeout{** loaded for the language `#1'. Using the pattern for}%
\typeout{** the default language instead.}%
\else
\language=\csname l@#1\endcsname
\fi
#2}}
\providecommand{\BIBdecl}{\relax}
\BIBdecl

\bibitem{ref:largescale}
J.~Dean, G.~Corrado, R.~Monga, K.~Chen, M.~Devin, M.~Mao, A.~Senior, P.~Tucker,
  K.~Yang, Q.~V. Le \emph{et~al.}, ``Large scale distributed deep networks,''
  in \emph{Advances in neural information processing systems}, 2012, pp.
  1223--1231.

\bibitem{ref:tail}
J.~Dean and L.~A. Barroso, ``The tail at scale,'' \emph{Communications of the
  ACM}, vol.~56, no.~2, pp. 74--80, 2013.

\bibitem{ref:speedup}
K.~Lee, M.~Lam, R.~Pedarsani, D.~Papailiopoulos, and K.~Ramchandran, ``Speeding
  up distributed machine learning using codes,'' \emph{IEEE Transactions on
  Information Theory}, vol.~64, no.~3, pp. 1514--1529, 2018.

\bibitem{ref:highd}
K.~Lee, C.~Suh, and K.~Ramchandran, ``High-dimensional coded matrix
  multiplication,'' in \emph{Information Theory (ISIT), 2017 IEEE International
  Symposium on}.\hskip 1em plus 0.5em minus 0.4em\relax IEEE, 2017, pp.
  2418--2422.

\bibitem{ref:polycode}
Q.~Yu, M.~Maddah-Ali, and S.~Avestimehr, ``Polynomial codes: an optimal design
  for high-dimensional coded matrix multiplication,'' in \emph{Advances in
  Neural Information Processing Systems}, 2017, pp. 4403--4413.

\bibitem{ref:Tavor}
T.~Baharav, K.~Lee, O.~Ocal, and K.~Ramchandran, ``Straggler-proofing
  massive-scale distributed matrix multiplication with d-dimensional product
  codes,'' 2018.

\bibitem{ref:gradientcodingArxiv}
N.~Raviv, I.~Tamo, R.~Tandon, and A.~G. Dimakis, ``Gradient coding from cyclic
  mds codes and expander graphs,'' \emph{arXiv preprint arXiv:1707.03858},
  2017.

\bibitem{ref:gradientdescentRS}
R.~Tandon, Q.~Lei, A.~G. Dimakis, and N.~Karampatziakis, ``Gradient coding:
  Avoiding stragglers in distributed learning,'' in \emph{International
  Conference on Machine Learning}, 2017, pp. 3368--3376.

\bibitem{ref:codedconvolution}
S.~Dutta, V.~Cadambe, and P.~Grover, ``Coded convolution for parallel and
  distributed computing within a deadline,'' in \emph{Information Theory
  (ISIT), 2017 IEEE International Symposium on}, pp. 2403--2407.

\bibitem{ref:codedFourier}
Q.~Yu, M.~A. Maddah-Ali, and A.~S. Avestimehr, ``Coded fourier transform,'' in
  \emph{Communication, Control, and Computing (Allerton), 2017 55th Annual
  Allerton Conference on}.\hskip 1em plus 0.5em minus 0.4em\relax IEEE, 2017,
  pp. 494--501.

\bibitem{ref:shortdot}
S.~Dutta, V.~Cadambe, and P.~Grover, ``Short-dot: Computing large linear
  transforms distributedly using coded short dot products,'' in \emph{Advances
  In Neural Information Processing Systems}, 2016, pp. 2100--2108.

\bibitem{ref:matrixsparsif}
G.~Suh, K.~Lee, and C.~Suh, ``Matrix sparsification for coded matrix
  multiplication,'' in \emph{Communication, Control, and Computing (Allerton),
  2017 55th Annual Allerton Conference on}.\hskip 1em plus 0.5em minus
  0.4em\relax IEEE, 2017, pp. 1271--1278.

\bibitem{ref:HCMM}
A.~Reisizadeh, S.~Prakash, R.~Pedarsani, and A.~S. Avestimehr, ``Coded
  computation over heterogeneous clusters,'' \emph{arXiv preprint
  arXiv:1701.05973}, 2017.

\bibitem{ref:hier}
H.~Park, K.~Lee, J.-y. Sohn, C.~Suh, and J.~Moon, ``Hierarchical coding for
  distributed computing,'' \emph{arXiv preprint arXiv:1801.04686}, 2018.

\bibitem{ref:AmazonEC2}
\BIBentryALTinterwordspacing
 [Online]. Available: \url{https://aws.amazon.com/ec2/?nc1=h\_ls}
\BIBentrySTDinterwordspacing

\bibitem{ref:MapReduce}
J.~Dean and S.~Ghemawat, ``Mapreduce: simplified data processing on large
  clusters,'' \emph{Communications of the ACM}, vol.~51, no.~1, pp. 107--113,
  2008.

\bibitem{ref:ShuffleWatcher}
F.~Ahmad, S.~T. Chakradhar, A.~Raghunathan, and T.~Vijaykumar,
  ``Shufflewatcher: Shuffle-aware scheduling in multi-tenant mapreduce
  clusters.'' in \emph{USENIX Annual Technical Conference}, 2014, pp. 1--12.

\bibitem{ref:ScaleOut}
A.~Vahdat, M.~Al-Fares, N.~Farrington, R.~N. Mysore, G.~Porter, and
  S.~Radhakrishnan, ``Scale-out networking in the data center,'' \emph{Ieee
  Micro}, vol.~30, no.~4, pp. 29--41, 2010.

\bibitem{ref:HeteroEnv}
M.~Zaharia, A.~Konwinski, A.~D. Joseph, R.~H. Katz, and I.~Stoica, ``Improving
  mapreduce performance in heterogeneous environments.'' in \emph{Osdi},
  vol.~8, no.~4, 2008, p.~7.

\bibitem{ref:balance1}
W.~Halbawi, N.~Azizan, F.~Salehi, and B.~Hassibi, ``Improving distributed
  gradient descent using reed-solomon codes,'' in \emph{2018 IEEE International
  Symposium on Information Theory (ISIT)}.\hskip 1em plus 0.5em minus
  0.4em\relax IEEE, 2018, pp. 2027--2031.

\bibitem{ref:balance2}
W.~Halbawi, Z.~Liu, and B.~Hassibi, ``Balanced reed-solomon codes for all
  parameters,'' in \emph{Information Theory Workshop (ITW), 2016 IEEE}.\hskip
  1em plus 0.5em minus 0.4em\relax IEEE, 2016, pp. 409--413.

\bibitem{ref:aymptotic}
H.~A. David and H.~N. Nagaraja, ``Order statistics, hoboken,'' \emph{NJ: John
  Wiley \& Sons}, vol.~7, pp. 159--61, 2003.

\end{thebibliography}
\end{document}